\documentclass{custom} 

\title[Gaussian discrepancy]{Gaussian discrepancy: a probabilistic relaxation of vector balancing}
\usepackage{bbm}
\usepackage{enumitem}
\usepackage{graphicx}
\usepackage{mathtools}
\usepackage{microtype}
\usepackage{tikz}
\usepackage{times}
\usepackage[Symbolsmallscale]{upgreek}
\usepackage{thm-restate} 
\usepackage{style}
\usepackage{mdframed} 
\usepackage{amsmath}
\usepackage{euscript}

\makeatletter
\def\set@curr@file#1{\def\@curr@file{#1}} 
\makeatother
\usepackage[load-configurations=version-1]{siunitx} 

\usepackage{float}

\coltauthor{\Name{Sinho Chewi} \Email{schewi@mit.edu}\\
 \Name{Patrik Gerber} \Email{prgerber@mit.edu} \\
 \Name{Philippe Rigollet} \Email{rigollet@mit.edu}\\
 \addr MIT 
 \AND
 \Name{Paxton Turner} \Email{paxtonturner@fas.harvard.edu} \\
 \addr Harvard University
}

\begin{document}

\maketitle

\begin{abstract}

We introduce a novel relaxation of combinatorial discrepancy called \emph{Gaussian discrepancy}, whereby binary signings are replaced with correlated standard Gaussian random variables. This relaxation effectively reformulates an optimization problem over the Boolean hypercube into one over the space of correlation matrices. We show that Gaussian discrepancy is a tighter relaxation than the previously studied vector and spherical discrepancy problems, and we construct a fast online algorithm that achieves a version of the Banaszczyk bound for Gaussian discrepancy. This work also raises new questions such as the Koml\'{o}s conjecture for Gaussian discrepancy, which may shed light on classical discrepancy problems.

\end{abstract}




\section{Introduction and overview of the paper}

In this work we introduce a probabilistic relaxation of the classical combinatorial discrepancy problem that we call \emph{Gaussian discrepancy}. In this section, we first briefly survey discrepancy theory and formally define our relaxation. Then we discuss our main results which consist of (i) sharp comparisons of Gaussian discrepancy to previously studied relaxations, and (ii) a fast algorithm for online Gaussian discrepancy. We conclude this overview with some open problems and an outline of the remainder of the paper.  

\subsection{Discrepancy theory and Gaussian discrepancy}

\subsubsection{Background on combinatorial discrepancy}

Discrepancy theory is a rich area of mathematics which has both inspired the development of novel tools for its study, and found numerous applications in a variety of fields such as combinatorics, computer science, geometry, optimization, and statistics; see the textbooks~\citet{matousek, chazelle}.
In one of its most fundamental forms, discrepancy asks the following question: given an $m\times n$ matrix $A \in \R^{m\times n}$, determine the value of the \emph{discrepancy} of $A$, defined as
\begin{align}
\label{eq:disc}
    \disc(A)
    &\deq \min_{\sigma \in {\{\pm1\}}^n}{\norm{A\sigma}_\infty}\, = \min_{\sigma \in {\{\pm1\}}^n}{\max_{1 \leq i \leq m}{\abs{(A\sigma)_i}}}\,.
\end{align}
This question can be interpreted in terms of \emph{vector balancing}: if $v_1,\dotsc,v_n$ denote the columns of $A$, then we are looking for a \emph{signing}, that is, a vector of signs $\sigma \in {\{\pm1\}}^n$, that makes the signed sum $\sum_{j=1}^n \sigma_j v_j$ have small entries. 

A seminal result in this area, due to~\citet{spencer} and independently~\citet{gluskin1989extremal} (see also~\citet{giannopoulous1997vectorbalancing}), states that \begin{equation}
\label{eq:spencer}
    \disc(A) \le 6\sqrt n
\end{equation}
when $m=n$ and the entries of $A$ are bounded in magnitude by $1$. This remarkable result is the best possible, up to the constant factor $6$, and should be compared with the discrepancy incurred by a signing chosen uniformly at random which is of order $\Theta(\sqrt{n\log n})$. A far-reaching extension of this result is the \emph{Koml\'os conjecture} \citep{spencer}. 

\begin{conjecture}[Koml\'{o}s conjecture]
\label{conj:komlos}
There exists a constant $K > 0$ such that for any matrix $A$ whose columns have Euclidean norm at most $1$, it holds that
\[
\disc(A) \leq K\,. 
\]
\end{conjecture}
This conjecture remains one of the most important open problems in the field, and the best known bound for the Koml\'os problem, due to~\cite{banaszczyk}, yields $\disc(A) = O(\sqrt{\log (m \land n)})$.\footnote{Here $m \land n = \min\{m,n\}$.} The Koml\'{o}s conjecture contains as a special case the long-standing Beck--Fiala conjecture which states that if $A$ has $t$-sparse columns in $\{0, 1\}^m$, then $\disc(A) = O(\sqrt{t})$ \citep{BecFia81}. 

The original proofs of these results are non-constructive in the sense that they do not readily yield efficient algorithms for computing signings which achieve these discrepancy upper bounds. In the last decade, considerable effort was devoted to matching these upper bounds \textit{algorithmically}. Starting with the breakthrough work of~\citet{Ban10}, there are now a number of algorithmic results matching Spencer's bound ~\citep{LovMek12, HarSchSin14,LevRamRot17, Rot17, EldSin18}. The task of making Banaszczyk's bound algorithmic was more challenging, and it was settled in the last few years in a line of works~\citep{BanDadGar16, LevRamRot17, BanDadGarLov18, DadGarLov19}.

Recently, \emph{online} discrepancy minimization~\citep{ BanSpe20,bansal2020online,alweiss2020discrepancy, BanJiaMek21, LiuSahSawhney21} has seen increasing interest and has led to a new perspective on Banaszczyk's result.
In the \emph{oblivious} online setting, an adversary picks in advance vectors $v_1,\dotsc,v_T \in \R^m$, each with Euclidean norm at most $1$. During each round $t=1,\dotsc,T$, the algorithm receives the vector $v_t$, and it must output a sign $\sigma_t \in \{\pm1\}$. The goal of the algorithm is to minimize the maximum discrepancy incurred at any time, \textit{i.e.}\ the quantity $\max_{t \in [T]}{\norm{\sum_{s=1}^t \sigma_s v_s}_\infty}$. \citet{alweiss2020discrepancy} conjecture the following online version of the Banaszczyk bound.

\begin{conjecture}[Online Banaszczyk]
\label{conj:online_banaszczyk}
There exists a randomized algorithm for online balancing in the oblivious adversarial setting that with high probability achieves the bound
\[
\max_{t \in [T]}{\Bigl\lVert \sum_{s=1}^t \sigma_s v_s\Bigr\rVert_\infty} \lesssim \sqrt{ \log(mT)}\,,
\]
for any sequence $v_1, \ldots, v_T \in \R^m$ of vectors with Euclidean norm at most $1$.
\end{conjecture}
This would be nearly optimal, as a lower bound of $\widetilde \Omega(\sqrt{\log T})$ is established in~\citet{bansal2020online}. 

\subsubsection{Gaussian discrepancy and the coupling perspective}

Motivated by the aforementioned longstanding conjectures, recent algorithmic progress, and the goal of shedding new light on the classical discrepancy minimization problem, in this work we introduce a novel relaxation called \emph{Gaussian discrepancy}. Our route to Gaussian discrepancy is through an alternative perspective on the discrepancy objective \eqref{eq:disc} based on couplings of random variables.

Recall that a Rademacher random variable is distributed uniformly on $\{\pm1\}$. A \emph{coupling} of Rademacher random variables is a random vector for which each marginal distribution is Rademacher. Since a signing $\sigma \in \{\pm1\}^n$ and its negative $-\sigma $ achieve the same discrepancy, the uniform distribution on the set of optimal signings furnishes an \textit{optimal} Rademacher coupling that minimizes the right-hand-side of \eqref{eq:disc} in expectation. Precisely, it holds that
\begin{align}\label{eq:disc_as_coupling}
    \disc(A)
    = \min\Bigl\{\E\norm{A\sigma}_\infty : \Pr(\sigma_j = -1) = \Pr(\sigma_j = +1) = \frac{1}{2}~\text{for all}~j \in [n]\Bigr\}\, , 
\end{align}
where the minimization above is over couplings of Rademacher random variables.\footnote{Such a minimization can be seen as an instance of a multimarginal optimal transport problem~\citep{pass2015mot,AltBoi21}. }

The coupling perspective plays an important role in discrepancy theory and its applications. The recent algorithmic proof of Banaszczyk's theorem by \citet{BanDadGarLov18} relies on the equivalence between Banaszczyk's theorem\footnote{More precisely, in this statement we refer to the general result of~\cite{banaszczyk} about balancing vectors to lie in a convex body $K$; the case where $K$ is the scaled $\ell_\infty$ ball corresponds to the combinatorial discrepancy.} and the existence of sub-Gaussian distributions supported on $\{ \sum_{j=1}^n \sigma_j v_j \}_{\sigma \in \{\pm1\}^n}$, as established by \citet{DadGarLov19} . Their algorithm, known as the \emph{Gram-Schmidt walk}, focuses on correlating the entries of $\sigma$ in order to control the sub-Gaussian constant of $A\sigma$. Further, in applications of discrepancy theory such as randomized control trials, it is important to output not only a single signing but rather a \emph{distribution} over low-discrepancy signings for the purpose of inferring treatment effects~\citep{KriAzrKap19,TurMekRig20, HarSavSpiZha19}. 

To construct our relaxation, we replace the Rademacher distribution with the standard Gaussian distribution in the coupling interpretation \eqref{eq:disc_as_coupling} of combinatorial discrepancy. The \emph{Gaussian discrepancy} is defined to be
\begin{align}
\label{eq:gdisc}
    \discG(A)
    = \min\Bigl\{\mathbb{E}\norm{Ag}_\infty : g_j \sim \normal(0, 1)~\text{for all}~j\in [n] \,, \; g_1,\dotsc,g_n~\text{jointly Gaussian}\Bigr\}\,, 
\end{align}
where $\NN(0,1)$ denotes the standard normal distribution. Equivalently, the minimization is over covariance matrices $\Sigma$ for the random vector $g = (g_1, \ldots, g_n)$ that lie in the \emph{elliptope} $\mc{E}_n$, the set of all positive semidefinite matrices whose diagonal is the all-ones vector. Unlike Rademacher couplings which require, a priori, an exponential number of parameters to describe, joint Gaussian couplings admit a compact description that is completely determined by the covariance matrix $\Sigma \in \R^{n \times n}$.

Let us record the simple observation that Gaussian discrepancy is indeed a relaxation of combinatorial discrepancy.

\begin{proposition}\label{prop:discG_is_relaxation}
    For any matrix $A$, it holds that
    $\discG(A) 
    \leq \sqrt{2/\pi}  \disc(A)\,.$
\end{proposition}
\begin{proof}
    Given an optimal signing $\sigma$ for $\disc(A)$, let $g_j = \sigma_j \xi$ for $j \in [n]$, where $\xi$ is a standard Gaussian (equivalently, the covariance matrix of $g$ is $\Sigma = \sigma\sigma^\T$).
    Then,
    \begin{align*}
        \discG(A)
        &\le \E\norm{Ag}_\infty
        = \norm{A\sigma}_\infty \E\abs{\xi}
        = \sqrt{2/\pi} \disc(A)\,.
    \end{align*}
\end{proof}
As a consequence, results on ordinary discrepancy, such as Spencer's theorem and Banaszczyk's theorem, immediately translate into bounds on the Gaussian discrepancy. 

\subsection{Notation}

We write $x \land y = \min\{x,y\}$, $x \lor y = \max\{x,y\}$, and $[n]=\{1, 2, \dotsc, n\}$. We write $\NN(\mu, \Sigma)$ for the Gaussian distribution with mean $\mu$ and covariance matrix $\Sigma$. $S^{n-1}$ denotes the unit sphere centered at the origin in $\R^n$, and $\mathbb S^n_+$ denotes the set of symmetric positive semidefinite $n\times n$ matrices. For $A,B \in \mathbb S^n_+$ we write $A \preceq B$ if $B-A\in\mathbb S^n_+$. For vectors $x \in \R^n$ we write $\norm{x}_p,\,p\in[1,\infty]$ for the $\ell_p$ norm of $x$, while for matrices $X$, $\norm{X}_{p \to q}$ denotes the induced operator norm from $\ell_p$ to $\ell_q$. We write $\langle \cdot, \cdot \rangle$ for both the Euclidean inner product between vectors and the Frobenius inner product between matrices. For two sequences of positive real numbers $(a_n), (b_n)$ we write $a_n \lesssim b_n$ if there exists a constant $C>0$ with $a_n \leq C b_n$ for all sufficiently large $n$, and write $a_n \asymp b_n$ if $a_n \lesssim b_n \lesssim a_n$. Thus, $a_n \lesssim b_n$ is a synonym for $a_n = O(b_n)$. The $n \times n$ identity matrix is denoted by $I_n$, and the all-ones vector is denoted by $\mb 1_n \in \R^n$. We define the elliptope to be $\mc{E}_n = \{ \Sigma \in \R^{n\times n}: \Sigma \succeq 0, \; \mathsf{diag}(\Sigma) = \mb 1_n \}$.

\subsection{Results}
\label{sscn:results}
\subsubsection{Comparisons between relaxations of discrepancy}
\label{scn:relaxations} 
\label{sscn:comparison_results}
In this section we develop an understanding of Gaussian discrepancy by comparing it to the \emph{vector discrepancy} and \emph{spherical discrepancy} relaxations of combinatorial discrepancy. Let us first define these notions and describe some results known about them.

The \emph{vector discrepancy} relaxation replaces the signs $\sigma_1, \ldots, \sigma_n \in \{\pm 1\}$ in combinatorial discrepancy \eqref{eq:disc} with unit vectors $u_1, \ldots, u_n \in S^{n-1}$. Formally,
\begin{align}
\label{eq:vdisc}
    \discv(A)
    &\deq \min\Bigl\{\max_{i\in [m]}{\Bigl\lVert \sum_{j=1}^n A_{i,j} u_j\Bigr\rVert_2} : u_1,\dotsc,u_n\in S^{n-1}\Bigr\}\,.
\end{align}
This problem is recast as a semidefinite program over the elliptope by constructing the Gram matrix, $\Sigma_{i, j} \deq \langle u_i, u_j \rangle$ for $i,j \in [n]$, of the unit vectors $u_1, \ldots, u_n$; see \citet{Niko13} for more details. 

Given a signing $\sigma \in \{ \pm 1 \}^n$ and $u \in S^{n-1}$, we can associate to it the unit vectors $u_j = \sigma_j u,\,j\in[n]$. From this we see that vector discrepancy is indeed a relaxation of discrepancy: \[\discv(A) \le \disc(A)\,.\] 
Vector discrepancy has been highly influential in discrepancy theory. It led to the initial algorithm of~\citet{Ban10} for Spencer's theorem, which uses a random walk guided by vector discrepancy solutions \citep[see also][]{BanDadGar16}, as well as the recent constructive proof of Banaszczyk's theorem \citep{BanDadGarLov18}.\footnote{\citet{BanDadGarLov18} remark that their Gram-Schmidt algorithm was ``\ldots inspired by the constructive proof of~\cite{DadGarLov19} for the existence of solutions to the Koml\'os vector coloring program of Nikolov."} Vector discrepancy has also been studied in its own right in the work of \citet{Niko13} which provides an in-depth analysis of this relaxation using SDP duality. 

\emph{Spherical discrepancy} was more recently introduced by \citet{jonesmcpartlon2020spheericaldisc}. It is obtained by relaxing the space of solutions in~\eqref{eq:disc} from ${\{\pm1\}}^n$ to the sphere of radius $\sqrt{n}$. Formally,
\begin{align}
\label{eq:sdisc}
    \discs(A)
    &\deq \min_{x \in \sqrt n \, S^{n-1}}{\norm{Ax}_\infty}\,.
\end{align}
\citet{jonesmcpartlon2020spheericaldisc} prove sharp bounds on spherical discrepancy in the setting of  Spencer's theorem ($\abs{A_{i,j}} \leq 1$ for all $i, j$) and the Koml\'{o}s conjecture ($\norm{A_{:,j}}_2 \leq 1$ for all $j \in [n]$) and apply their results to derive lower bounds for certain sphere covering problems. 

With these definitions in hand, it is natural to wonder about the relationships between the various relaxations of discrepancy. For example, which relaxation gives the best approximation to combinatorial discrepancy? Our first result, whose proof we defer to Section \ref{scn:further_comparisons}, provides a unifying probabilistic perspective on the three relaxations of discrepancy and leads to a straightforward comparison.

\begin{theorem}\label{thm:sdp_spherical_are_relaxations}
    For any matrix $A$, it holds that
    \begin{alignat}{3}
        \discv(A)
        &\asymp \min\bigl\{&&\max_{i\in [m]} \E\abs{{(Ag)}_i} &&\bigm\vert g \sim \normal(0, \Sigma)\,, \;\Sigma \in \mbb S_+^n\,, \;\diag \Sigma = \mb 1_n\bigr\}\,, \label{eq:vdisc_prob} 
        \\
        \discs(A)
        &\asymp \min\bigl\{ &&\E\max_{i\in [m]}{\abs{{(Ag)}_i}} &&\bigm\vert g \sim \normal(0, \Sigma)\,, \;\Sigma \in \mbb S_+^n\,,\; \tr \Sigma = n\bigr\}\,. \label{eq:sdisc_prob}
    \end{alignat}
    Here, $\mbb S_+^n$ is the set of symmetric positive semidefinite $n\times n$ matrices.
\end{theorem}

For comparison, we recall that
\begin{align*}
    \discG(A)
    &= \min\bigl\{\E\max_{i\in [m]}{\abs{{(Ag)}_i}} \bigm\vert g \sim \normal(0, \Sigma)\,, \;\Sigma \in \mbb S_+^n\,,\; \diag \Sigma = \mb 1_n\bigr\}\,.
\end{align*}
Hence, $\discv$ relaxes the objective of Gaussian discrepancy by moving the maximum over rows outside of the expectation, whereas $\discs$ relaxes the constraint of Gaussian discrepancy from $\diag \Sigma = \mb 1_n$ to $\tr \Sigma = n$.
In comparison, from Proposition~\ref{prop:discG_is_relaxation}, the usual definition of discrepancy can be understood as \emph{adding} a constraint to Gaussian discrepancy, namely that $\rank \Sigma = 1$. The following relationship between the notions of discrepancy is an immediate consequence.

\begin{corollary}\label{cor:comparison_between_disc}
    For any matrix $A$,
    \begin{align*}
        \discv(A) \vee \discs(A)
        &\lesssim \discG(A)
        \lesssim \disc(A)\,.
    \end{align*}
\end{corollary}
In words, Gaussian discrepancy is a tighter relaxation of discrepancy than vector discrepancy and spherical discrepancy. Moreover, we show via a suite of examples in Section \ref{scn:further_comparisons} that none of the inequalities in Corollary \ref{cor:comparison_between_disc} can be reversed up to constant factor and that spherical discrepancy and vector discrepancy are incomparable in general. 

Although Gaussian discrepancy is always larger than the vector discrepancy, the two relaxations are in fact closely related. Indeed, their common feasible set is the elliptope, which consists of the Gram matrices of collections of unit vectors or equivalently the covariance structures of feasible Gaussian couplings. The connection between these two relaxations yields our main tools for bounding Gaussian discrepancy as well as an $O(\sqrt{\log m})$-factor approximation algorithm for computing Gaussian discrepancy.

Our main inequality for controlling Gaussian discrepancy relies on a notion of \textit{rank-constrained vector discrepancy}, defined as follows. For every rank $r \in [n]$, let
\begin{align}\label{eq:vdisc_r}
    \discv_r(A)
    &\deq \min\Bigl\{\max_{i\in [m]}{\Bigl\lVert \sum_{j=1}^n A_{i,j} u_j\Bigr\rVert_2} : u_1,\dotsc,u_n\in S^{r-1}\Bigr\}\,.
\end{align}
Equivalently we require the Gram matrix $\Sigma \deq {(\langle u_i, u_j \rangle)}_{i,j\in [n]}$ to have rank at most $r$. Since a rank-$1$ matrix $\Sigma$ corresponds precisely to a signing, note that
\begin{align*}
    \discv(A) = \discv_n(A)
    &\le \discv_{n-1}(A) \le \cdots \le \discv_1(A) = \disc(A)\,.
\end{align*} 
The next result compares Gaussian discrepancy with this rank-constrained problem and is a crucial ingredient in our study of online Gaussian discrepancy.

\begin{proposition}\label{prop:discG_low_rank}
    For any $r\in [n]$, it holds that $\discG(A) \le \sqrt r \discv_r(A).$
\end{proposition}

\begin{proof}
    Let $u_1,\dotsc,u_n \in S^{r-1}$ be an optimal solution for $\SDP_r(A)$. Then, if $\xi$ denotes a standard Gaussian vector in $\R^r$, we can define $g_j \deq \langle u_j, \xi \rangle$ for $j \in [n]$, and it can be easily checked that this is feasible for the definition of Gaussian discrepancy.
    Hence,
    \begin{align}
    \label{eqn:low_rank_useful}
        \discG(A) 
        \le \E\max_{i\in [m]}{\Bigl\lvert \Bigl\langle \sum_{j=1}^n A_{i,j} u_j, \xi \Bigr\rangle \Bigr\rvert}
        \le \max_{i\in [m]}{\Bigl\lVert \sum_{j=1}^n A_{i,j} u_j \Bigr\rVert_2}  \E\norm{\xi}_2 
        \le \sqrt r \SDP_r(A)\,.
    \end{align}
    The last inequality uses the standard probabilistic fact that $\E\norm\xi_2 \le \sqrt r$.
    When $r=1$, this can be improved to $\E\abs\xi = \sqrt{2/\pi}$, which recovers Proposition~\ref{prop:discG_is_relaxation}.
\end{proof}

Applying a similar strategy but using the union bound instead of Cauchy--Schwarz shows that vector discrepancy upper bounds Gaussian discrepancy up to a logarithmic term.

\begin{proposition}\label{prop:discG_union_bound}
    It holds that 
    $\discG(A) \lesssim \sqrt{\log (2m)} \, \discv(A).$
\end{proposition}

\begin{proof}
    Let $u_1,\dotsc,u_n\in\R^n$ be an optimal solution for $\SDP_n(A)$, and let $\xi\sim\NN(0, I_n)$. As in the proof of Proposition~\ref{prop:discG_low_rank},
    \begin{align*}
        \discG(A)
        &\le \E\max_{i\in [m]}{\Bigl\lvert \Bigl\langle \sum_{j=1}^n A_{i,j} u_j, \xi \Bigr\rangle \Bigr\rvert}\,.
    \end{align*}
    The random variable $\langle \sum_{j=1}^n A_{i,j} u_j, \xi \rangle$ is sub-Gaussian with parameter $\norm{\sum_{j=1}^n A_{i,j} u_j}_2^2$. Thus the standard maximal inequality for sub-Gaussian random variables \citep[see, \textit{e.g.},][Theorem 2.5]{boucheronlugosimassart2013concentration} yields
    \begin{align}
    \label{eqn:union_bound_useful}
        \discG(A)
        &\le \sqrt{2\ln(2m)} \max_{i\in [m]}{\Bigl\lVert \sum_{j=1}^n A_{i,j} u_j \Bigr\rVert_2}
        = \sqrt{2\ln(2m)} \SDP_n(A)\,.
    \end{align}
\end{proof} 

We record three useful consequences of Proposition \ref{prop:discG_union_bound}. When combined with Theorem \ref{thm:sdp_spherical_are_relaxations} it implies that Gaussian discrepancy is approximated by vector discrepancy up to a logarithmic factor. Since vector discrepancy admits an SDP formulation, it can be computed in polynomial time using methods from convex optimization \citep{nesterov2018cvxopt}, and in turn this yields our claimed $O(\sqrt{\log m})$-factor approximation algorithm for Gaussian discrepancy.

Second, Proposition \ref{prop:discG_union_bound} allows us to draw a new connection between spherical discrepancy and vector discrepancy. Combining this result with Corollary~\ref{cor:comparison_between_disc} and  Proposition \ref{prop:discG_low_rank}, we obtain the following. 
\begin{corollary}\label{cor:spherical_and_sdp}
    For any matrix $A$, it holds that 
    $\discs(A) \lesssim \sqrt{n \land \log m} \discv(A).$
\end{corollary} 
As described in detail in Section \ref{scn:further_comparisons}, this inequality is tight and also may not be reversed --- in fact it is possible to have $\discs(A) = 0$ while $\discv(A) > 0$ (see Example \ref{ex:zero_sdisc_nonzero_vdisc}). We remark that this sharp result falls out naturally with Gaussian discrepancy as a mediator between spherical and vector discrepancy.

Finally, Proposition \ref{prop:discG_union_bound} leads to a simple algorithm that achieves Banaszczyk's bound for Gaussian discrepancy. The main result of \citet{Niko13} states that the Koml\'{o}s conjecture holds for vector discrepancy. Hence, a solution to the SDP formulation of vector discrepancy yields a feasible covariance matrix for Gaussian discrepancy with objective value $O(\sqrt{ \log m  })$.\footnote{This can be further refined to $O(\sqrt{ \log(m \wedge n)})$ using standard reductions.} Note that since Gaussian discrepancy is a relaxation, existing approaches for combinatorial discrepancy also achieve Banaszczyk's bound for Gaussian discrepancy with much more involved algorithms \citep{BanDadGar16,LevRamRot17,BanDadGarLov18}.

This raises the questions of whether or not the Koml\'{o}s conjecture (Conjecture \ref{conj:komlos}) or online Banaszczyk (Conjecture \ref{conj:online_banaszczyk}) can be proved for Gaussian discrepancy. More broadly, it is natural to ask whether unresolved conjectures about combinatorial discrepancy can be solved for a given relaxation. Indeed, \citet{Niko13} and \citet{jonesmcpartlon2020spheericaldisc} establish the Koml\'{o}s conjecture for vector and spherical discrepancy, respectively. While we are not able to establish a Gaussian version of the Koml\'{o}s conjecture, it serves as a tantalizing question for further study; see Section~\ref{scn:open_problems} for further discussion. On the other hand, our main algorithmic result establishes a Gaussian discrepancy variant of the online Banaszczyk bound, as discussed in the next section. 

\subsubsection{Online Banaszczyk bound for Gaussian discrepancy}
\label{sscn:online_disc}

We begin by recalling the setting of online discrepancy minimization with an oblivious adversary. First, an adversary picks vectors $v_1, \dots, v_T$ of $\ell_2$ norm at most $1$. Then, in each round $t=1,2,\dots,T$ the vector $v_t$ is revealed to the algorithm, which must output a sign $\sigma_t \in \{\pm 1\}$. The aim of the algorithm is to minimize the maximum discrepancy incurred, i.e.\ $\max_{t\in [T]}{\norm{\sum_{s=1}^t \sigma_s v_s}_\infty}$. Equivalently, in each round the algorithm is required to choose a \textit{vector} of signs $\sigma^{(t)} \in \{\pm1\}^t$ that obeys the following  \emph{consistency condition}:
\begin{equation}
\label{eq:ordinary_consistency}
\sigma^{(t)}_{[t-1]} = \sigma^{(t-1)}\,,
\end{equation}
where $x_{[t-1]}$ denotes the restriction of $x \in \R^t$ to the first $t - 1$ coordinates. That is, the new signing $\sigma^{(t)} \in \{\pm1\}^t$ chosen in round $t$ is forbidden from changing the signs specified in previous rounds.

The latter formulation motivates our Gaussian discrepancy variant of the online discrepancy problem. At each round $t$, we require the algorithm to output a $t$-dimensional correlation matrix $\Sigma^{(t)}$ which satisfies a consistency condition analogous to \eqref{eq:ordinary_consistency}, namely
\begin{equation}\label{eqn:gaussian_consistency}
    \Sigma^{(t)}_{[t-1]\times[t-1]} = \Sigma^{(t-1)}\,.
\end{equation}
Thus after round $t$, the correlations between the first $t$ coordinates must remain fixed. The aim is still to minimize the maximum discrepancy incurred, i.e.\ $\max_{t\in [T]}\E\norm{\sum_{s=1}^t g_s v_s}_\infty$, where $g=(g_1, \dotsc, g_T) \sim \normal(0, \Sigma^{(T)})$. A simple argument based on the Cholesky decomposition (see Appendix~\ref{scn:formulations equiv} for details) shows that any such sequence $\Sigma^{(1)}, \dotsc, \Sigma^{(T)}$ can be realized as a sequence of Gram matrices of initial segments of a sequence of unit vectors $u_1, \dotsc, u_T \in \ell_2$.
Hence, it is equivalent to require that the algorithm output a unit vector $u_t \in \ell_2$ at each round $t$, and then set $\Sigma^{(t)}_{i,j} = \langle u_i, u_j\rangle_{\ell_2}$. We arrive at the following formulation of online Gaussian discrepancy. As in online combinatorial discrepancy, we allow the algorithm to be randomized. 

\begin{problem}[Online Gaussian discrepancy with oblivious adversary]
\label{prob:online_gaussian_discrepancy}
An adversary selects vectors $v_1, \dotsc, v_T$ with $\ell_2$ norm at most $1$ in advance. At each round $t \in [T]$, the algorithm observes $v_t$ and outputs a random unit vector $u_t \in \ell_2$. The algorithm aims to minimize 
\begin{equation*}
    \discG( v_1, \ldots, v_T; \Sigma\rp{T} ) = \max\limits_{t \in [T]}  \E \Bigl\lVert\sum\limits_{s=1}^t g_s v_s\Bigr\rVert_\infty, \quad g \sim \NN(0, \Sigma\rp{T})
\end{equation*}
with high probability over the random correlation matrix $\Sigma_{i,j}^{(T)} = \langle u_i, u_j \rangle_{\ell_2}$. 
\end{problem}

We show in Appendix~\ref{scn:formulations equiv} that, without loss of generality, it suffices to consider $u_t$ whose support is a subset of the first $t$ coordinates.

One strategy for generating feasible couplings in each round is to first fix a \textit{rank} parameter $r \geq 1$ in advance and output a unit vector $u_t \in S^{r-1}$ in each round. Our main result is an algorithm of this form that solves the Gaussian discrepancy variant of the online Banaszczyk conjecture (Conjecture~\ref{conj:online_banaszczyk}). 

\begin{theorem}[Online Banaszczyk bound for Gaussian discrepancy]
\label{thm:online_gauss_banaszczyk} 
Let $v_1, \ldots, v_T \in \R^m$ denote vectors of $\ell_2$ norm at most $1$ selected in advance by an adversary. For all positive integers $r \geq 2$, there is a randomized online algorithm that with probability at least $1 - \delta$ outputs $u_1, \ldots, u_T \in S^{r-1}$ such that
\[
\discG( v_1, \ldots, v_T; \Sigma\rp{T} ) = O\bigl(\sqrt{\log(mT/\delta)}\bigr),
\]
where $\Sigma\rp{T}_{i,j} = \langle u_i, u_j \rangle$. The algorithm runs in time $O(mr)$ per round.
\end{theorem}

The proof of Theorem \ref{thm:online_gauss_banaszczyk} is the main content of Section \ref{scn:banaszczyk}. If we could prove this theorem with $r = 1$, then the online Banaszczyk problem (Conjecture~\ref{conj:online_banaszczyk}) would follow from the proof of Proposition~\ref{prop:discG_is_relaxation}. Unfortunately, there is an obstruction preventing us from considering $r=1$, as described further below. Nevertheless, Theorem \ref{thm:online_gauss_banaszczyk} shows that we are `one rank away' from establishing Conjecture~\ref{conj:online_banaszczyk}. 

Our algorithm is based on an intriguing idea in the recent paper~\citet{LiuSahSawhney21}: namely, if one can find a Markov chain on $\R$ whose increments take values in $\{\pm 1 \}$ and whose stationary distribution is a Gaussian, then it is possible to construct an algorithm, which they call the \emph{Gaussian fixed-point walk}, with the property that each partial sum $\sum_{s=1}^t \sigma_s v_s$ is the difference of two Gaussian vectors. The online Banaszczyk bound would then follow from a union bound. However, \citet{LiuSahSawhney21} exhibit a \textit{parity obstruction} which implies no such Markov chain exists on $\R$, and this leads them to consider instead Markov chains whose increments lie in $\{0, \pm 1\}$ or $\{\pm 1, 2\}$ (and hence the resulting algorithms output partial colorings or improper colorings). We show that an analogous Markov chain \emph{does} exist on $\R^r$ for any $r \geq 2$ whose increments lie in $S^{r-1}$ and whose stationary distribution is a Gaussian with independent and identically distributed coordinates (see Figure \ref{fig:Markov_chain}). Working with higher ranks $r \geq 2$ avoids certain technical complications resulting in the partial coloring version derived by \citet{LiuSahSawhney21} and allows for a simple algorithm and analysis.
\begin{figure}[H]
\centering
\subfigure[]{%
  \includegraphics[width=0.4\textwidth]{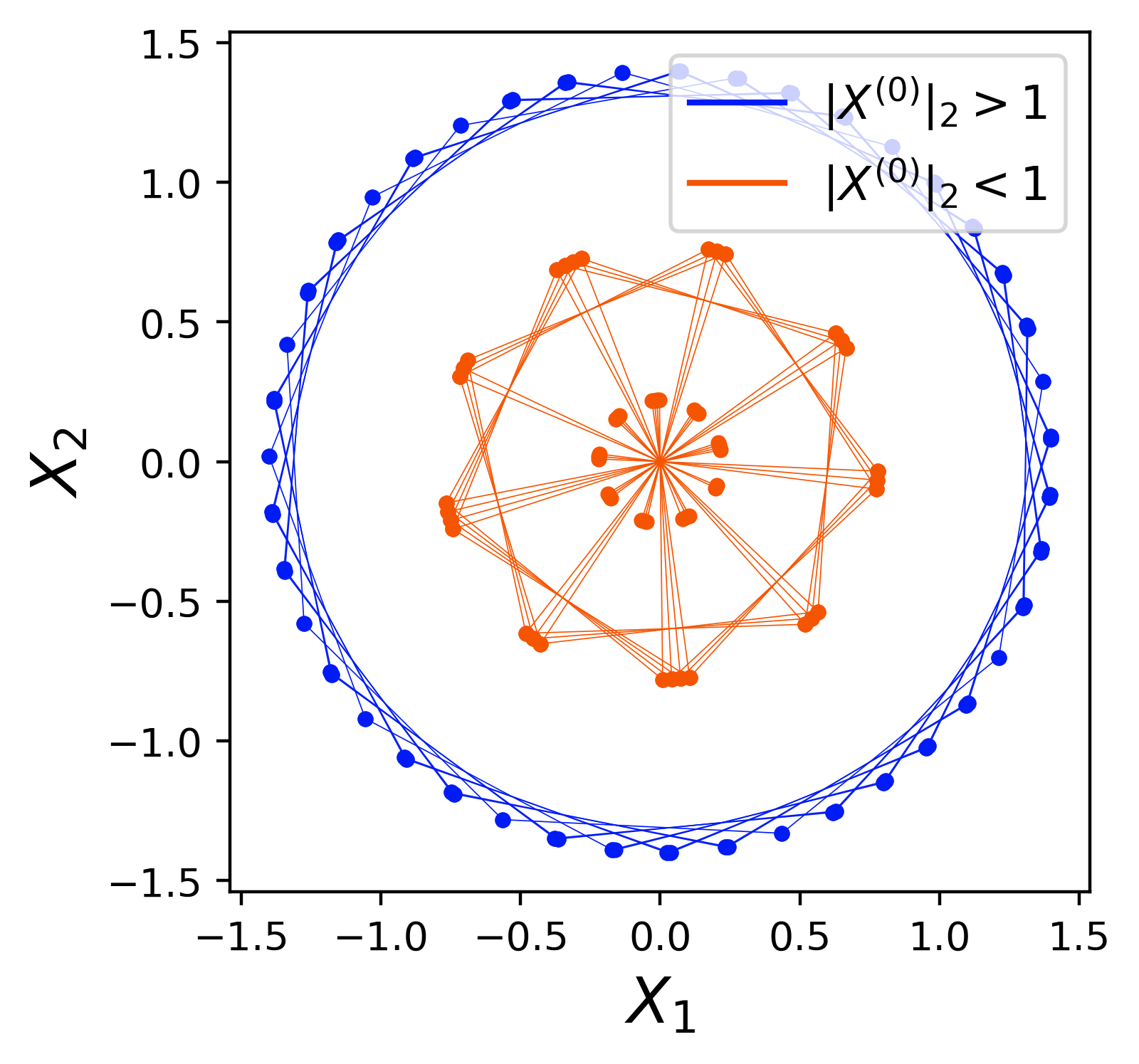}
  \label{fig:subfigure1}}
\quad
\hfill
\subfigure[]{%
  \includegraphics[width=0.4\textwidth]{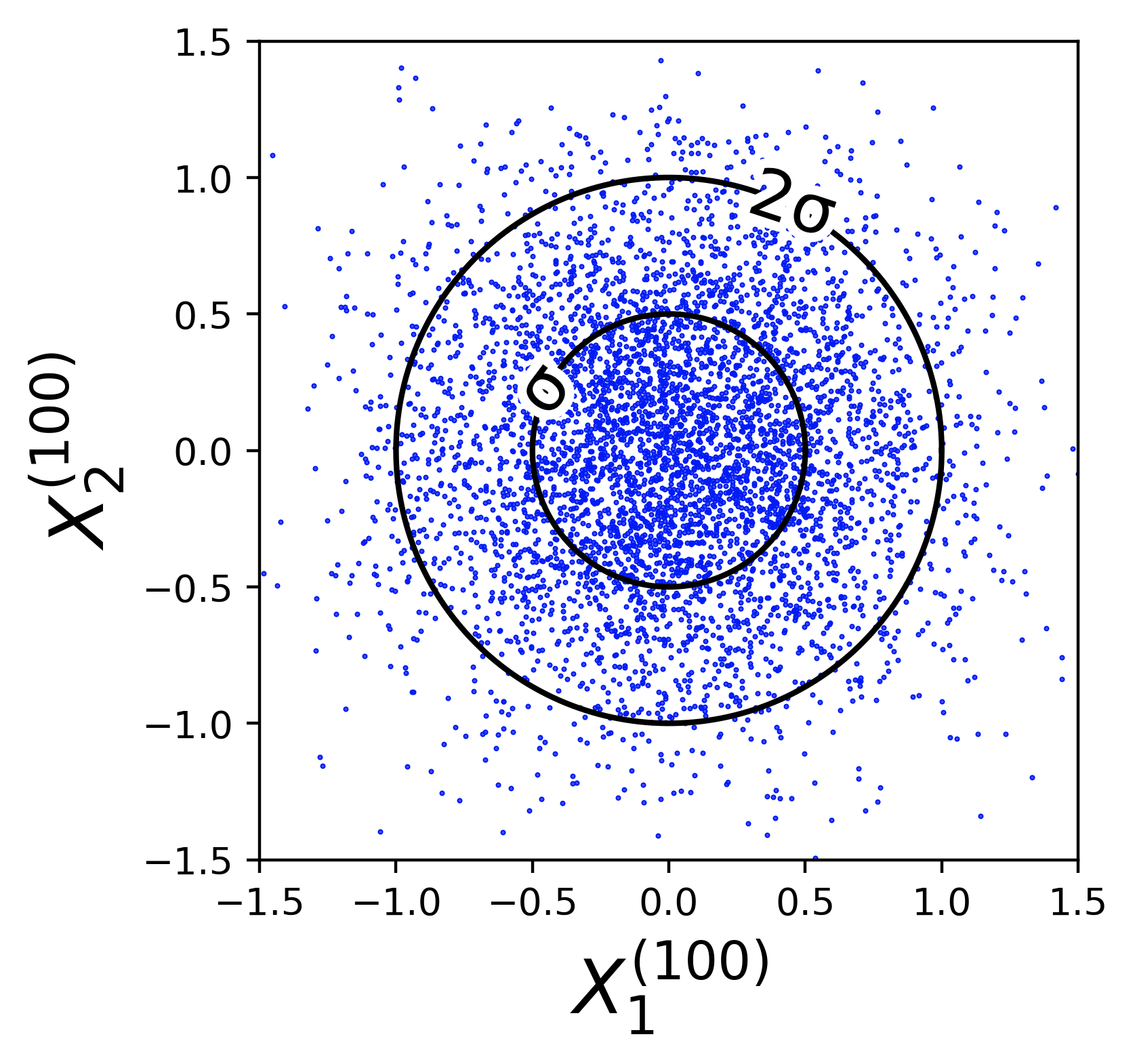}
  \label{fig:subfigure2}}
\caption{${(X\rp{t})}_{t \geq 0}$ is the Markov chain on $\R^{2}$ with unit vector steps and Gaussian stationary distribution. (a) Evolution of ${(X\rp{t})}_{t\ge 0}$ over two runs of the Markov chain with $2000$ steps. The orange and blue trajectories are initialized inside and outside the unit circle, respectively. (b) Scatter plot of $X\rp{100}$ over $5000$ independent runs of the Markov chain started from the stationary distribution $X\rp{0} \sim \NN(0, \sigma^2 I_2)$, where $\sigma=0.5$.
}
\label{fig:Markov_chain}
\end{figure}

We note that the guarantee of Theorem \ref{thm:online_gauss_banaszczyk} can be shown to be tight by using standard Gaussian estimates. In particular, taking $r >2$ in our algorithm provably does not improve the Gaussian discrepancy.

Interestingly, our algorithm also has consequences for the previously unexplored problem of online vector discrepancy, and for this problem, choosing $r > 2$ offers a substantial improvement. Perhaps surprisingly, Theorem \ref{thm:online_vector_komlos} below shows that the Koml\'{o}s conjecture, with sharp constant, is attainable for vector discrepancy even in the oblivious online setting. In contrast, \cite{bansal2020online} exhibit an $\widetilde \Omega(\sqrt{\log T})$ lower bound for online combinatorial discrepancy in the oblivious setting. The next result is an immediate consequence of Theorem \ref{thm:analysis_of_alg} proved in Section \ref{scn:banaszczyk}.

\begin{theorem}[Online Koml\'{o}s bound for vector discrepancy]
\label{thm:online_vector_komlos} 
Let $v_1, \ldots, v_T \in \R^m$ denote vectors of $\ell_2$ norm at most $1$ selected in advance by an adversary. Fix $\delta, \varepsilon > 0$.
    When run with rank $r = \Theta(\varepsilon^{-2}\log(mT/\delta))$, the algorithm outputs $u_1,\dotsc,u_T\in S^{r-1}$ online with 
    \[
    \max_{t\in [T]} \, \max_{i\in [m] }{\Bigl\lVert \sum_{s=1}^t {(v_{s})}_i \, u_s\Bigr\rVert_2} \le 1 + \varepsilon 
    \]
    with probability at least $1-\delta$. The algorithm runs in time $O(m \varepsilon^{-2} \log(mT/\delta))$ per round.
\end{theorem}

As a corollary of the above theorem, we obtain a new proof of the vector Koml\'os theorem of~\citet{Niko13}, which states that $\discv(A) \leq 1$ for all matrices $A$ whose columns have $\ell_2$ norm at most $1$.
Since the vector discrepancy of the identity matrix $I_T$ is $1$, our result is essentially sharp.

Although we present our algorithm for the online setting, it yields new algorithmic implications for the offline setting as well. In particular, our runtime is nearly linear in the input size $mT$ for moderate values of the approximation parameter $\eps$. This improves on the time complexity of previous algorithms for vector Koml\'{o}s such as off-the-shelf SDP solvers, which require $\widetilde O({(m \lor T)}^{3.5} \log(1/\varepsilon))$ arithmetic operations~\citep{nesterov2018cvxopt}, and an iterative approach of \citet{DadGarLov19} that na\"{i}vely runs in time $O(T^3\, (m \lor  T))$ due to costly matrix inversion steps.\footnote{We do not take into account fast matrix multiplication for these runtime estimates. We also note that an SDP solver would in general yield the stronger guarantee $\norm{\sum_{s=1}^T {(v_t)}_i \, u_t}_{2}^2 \le {\discv(A)}^2 + \varepsilon$ for $1 \leq i \leq m$.}

Furthermore, setting $\eps = 1$ in Theorem \ref{thm:online_vector_komlos} implies that the Koml\'{o}s conjecture holds for rank-constrained vector discrepancy with $r= O(\log(mT))$. To the best of our knowledge this result was previously unknown even in an offline setting. 

\subsection{Open problems}
\label{scn:open_problems}
Our work leads to some natural open questions which we briefly describe below.

\begin{enumerate}
    
    \item \textbf{Komlós conjecture for Gaussian discrepancy}. Does there exist a constant $K_{\rm G} > 0$ such that $\discG(A) \leq K_{\rm G}$ for any matrix $A$ whose columns have $\ell_2$ norm at most $1$? 
    
    Since Gaussian discrepancy is a relaxation, solving this conjecture is a natural prerequisite for establishing the original Koml\'{o}s conjecture (Conjecture \ref{conj:komlos}). It also leads to a related problem: prove that the Koml\'{o}s conjecture for Gaussian discrepancy implies the Koml\'{o}s conjecture. 
    
    \item \textbf{Rounding Gaussian discrepancy solutions.} Does there exist an efficient rounding scheme to convert Gaussian discrepancy solutions to low-discrepancy signings? For example, in the Koml\'{o}s setting,  does there exist a polynomial-time computable function $f$ from covariance matrices to signings, such that 
    \[
    \norm{A f(\Sigma)}_\infty \lesssim \E_{g \sim \mc{N}(0, \Sigma) }\norm{Ag}_\infty + O(1) \, \, ,
    \]
    for all matrices $A$ with $\max_{j\in [n]}{\norm{A_{:,j}}_2} \leq 1$? 
    
    
    In Appendix~\ref{scn:failure_of_rounding}, we show that two simple rounding schemes (Goemans--Williamson rounding and a PCA-based rounding) are not effective in the setting of Spencer's theorem and the Koml\'{o}s conjecture. 
    
    \item \textbf{Computational tractability of Gaussian discrepancy}. Given a matrix, can we compute its Gaussian discrepancy exactly or approximately in polynomial time? In particular, is it NP-hard to approximate Gaussian discrepancy up to a constant factor? 
    We note that hardness of approximation results are already known for combinatorial discrepancy \citep{ChaNewNik11} and spherical discrepancy \citep{jonesmcpartlon2020spheericaldisc}.  
    
    \item \textbf{Achieving Banaszczyk's bound online}. Finally, we mention that the question of achieving Banaszczyk's bound online for combinatorial discrepancy, as originally posed by \citet{alweiss2020discrepancy}, is still open. 
\end{enumerate}

\subsection{Organization of the paper}

In Section~\ref{scn:main_comparison_pf}, we give the proof of our main comparison result between the relaxations of discrepancy (Theorem~\ref{thm:sdp_spherical_are_relaxations}). In Section~\ref{scn:tightness_examples}, we provide a suite of examples showing that the inequalities given in Section~\ref{sscn:comparison_results} are sharp. Finally, in Section~\ref{scn:banaszczyk}, we present our algorithm and prove that it achieves the online Banaszczyk bound for Gaussian discrepancy.

\section{Comparisons between relaxations of discrepancy}
\label{scn:further_comparisons}

In this section we prove Theorem~\ref{thm:sdp_spherical_are_relaxations} and give a number of illustrative examples that highlight the differences between the various relaxations of discrepancy. These examples show that none of the inequalities in Corollary~\ref{cor:comparison_between_disc} can be reversed in general (at least with constants independent of $m$ and $n$). Moreover, our results in this section give evidence for the following assertion: vector discrepancy and spherical discrepancy each capture distinct aspects of the original discrepancy problem, namely SDP solutions are ``aligned with the coordinate axes'', whereas spherical discrepancy solutions are ``low rank''. Gaussian discrepancy appears to capture both of these aspects simultaneously.

This assertion can already be partially justified by observing that the feasible covariance matrices $\Sigma$ for vector discrepancy satisfy $\diag \Sigma = \mb 1_n$, i.e.\ the variance along each of the coordinate axes is fixed to be $1$. On the other hand, from the proof of Theorem~\ref{thm:sdp_spherical_are_relaxations} below, we see that a spherical discrepancy solution $x$ gives rise to the covariance matrix $\Sigma \deq xx^\T$, which is indeed low rank but is not guaranteed to have unit variance along each coordinate axis.

\subsection{Proof of main comparison result}\label{scn:main_comparison_pf}

We repeatedly use the following useful facts about log-concave distributions. Note that any Gaussian random variable has a log-concave distribution.

\begin{lemma}[{{\citet[Propositions A.5-A.6]{alonsogutierrezbastero2015kls}}}]\label{lem:log_concave}
Let $X$ be a random vector in $\R^d$ with a log-concave distribution, and let $\norm\cdot : \R^d\to\R_+$ be any seminorm.
\begin{enumerate}
    \item For any $p > 1$, it holds that ${\E[\norm X^p]}^{1/p} \lesssim p \E\norm X$, where the implied constant is universal.
    \item For any $t > 0$, it holds that $\Pr\{\norm X \le t \E \norm X\} \lesssim t$, where the implied constant is universal.
\end{enumerate}
\end{lemma}

\begin{proof}[Proof of Theorem~\ref{thm:sdp_spherical_are_relaxations}] \underline{Vector discrepancy}: Given $\Sigma \in \mc{E}_n$, let $g^\Sigma$ denote a centered Gaussian vector with covariance matrix $\Sigma$, and write $UU^\T=\Sigma$ with $u_j$ denoting the rows of $U$. Observe that
\begin{align}\label{eq:sdp_useful}
    \max_{i\in [m]}{\Bigl\lVert \sum_{j=1}^n A_{i,j} u_j\Bigr\rVert_2}
    &= \sqrt{\max_{i\in [m]}{\langle A_{i,:}, \Sigma A_{i,:} \rangle}}\,.
\end{align}
    Therefore 
    \begin{align*}
        \discv(A)
        = \min_{\Sigma \in \mc{E}_n} \sqrt{\max_{i\in [m]}{\langle A_{i,:}, \Sigma A_{i,:} \rangle}}
        &= \min_{\Sigma \in \mc{E}_n} \sqrt{\max_{i\in [m]} \E\Bigl[\Bigl\lvert \sum_{j=1}^n A_{i,j} g^\Sigma_j \Bigr\rvert^2\Bigr]} \\
        &\overset{(\star)}{\asymp} \min_{\Sigma \in \mc{E}_n} \max_{i\in [m]} \E \Bigl\lvert \sum_{j=1}^n A_{i,j} g^\Sigma_j \Bigr\rvert
        = \min_{\Sigma \in \mc{E}_n} \max_{i\in [m]} \E\abs{{(Ag^\Sigma)}_i}\,,
    \end{align*}
    where $(\star)$ uses Lemma~\ref{lem:log_concave} and Cauchy--Schwarz.
    
    \underline{Spherical discrepancy}: Let us temporarily denote
    \begin{align*}
        \on{\widetilde{\msf{disc}_{\msf s}}}(A)
        &\deq \min\{\E\norm{Ag}_\infty : g \sim \normal(0, \Sigma)\,,\; \Sigma \in \mbb S_+^n\,,\; \tr \Sigma = n\}\,.
    \end{align*}
    Let $g$ be a Gaussian that achieves the minimum in the definition of $\on{\widetilde{\msf{disc}_{\msf s}}}(A)$.
    Then, by Markov's inequality,
    \begin{align}\label{eq:sph_discG_2}
        \Pr\{\norm{Ag}_\infty \ge 3 \on{\widetilde{\msf{disc}_{\msf s}}}(A)\}
        &\le \frac{1}{3}\,.
    \end{align}
    From Lemma~\ref{lem:log_concave}, we further have
    \begin{align}\label{eq:sph_discG_1}
        \Pr\{\norm g_2 \le c \E\norm g_2\} \le \frac{1}{3}\, \qquad\text{and}\qquad \E\norm g_2 \gtrsim \sqrt{\E\norm g_2^2}
    \end{align}
    for some constant $c>0$. From~\eqref{eq:sph_discG_1} and~\eqref{eq:sph_discG_2}, there exists a realization $\msf g \in \R^n$ of $g$ with
    \begin{equation*}
        \norm{A \msf{g}}_\infty
        \le 3 \on{\widetilde{\msf{disc}_{\msf s}}}(A)\,, \qquad\text{and}\qquad \norm{\msf{g}}_2 \geq c\E\norm{g}_2 \gtrsim \sqrt{\E\norm{g}_2^2} = \sqrt n\,.
    \end{equation*}
    It follows that
    \begin{align*}
        \discs(A)
        &\le \bigl\lVert A \, \frac{\sqrt n \, \msf{g}}{\norm{\msf{g}}_2} \bigr\rVert_\infty
        = \frac{\sqrt n}{\norm{\msf{g}}_2} \, \norm{A\msf{g}}_\infty
        \lesssim \on{\widetilde{\msf{disc}_{\msf s}}}(A)\,.
    \end{align*}
    Conversely, let $x$ be an optimal solution for $\discs(A)$ and let $\xi$ be a standard Gaussian variable on $\R$, so $\xi x \sim \normal(0, xx^\T)$ and $\tr(xx^\T) = \norm x_2^2 = n$.
    Then,
    \begin{align*}
        \on{\widetilde{\msf{disc}_{\msf s}}}(A)
        \le \E\norm{A(\xi x)}_\infty
        &= \E\abs \xi \, \norm{Ax}_\infty
        = \sqrt{\frac{2}{\pi}} \discs(A)\,,
    \end{align*}
    which establishes the converse bound. 
\end{proof}

\subsection{Tightness of the comparison results}\label{scn:tightness_examples}

Our first example demonstrates the sharpness of Corollary \ref{cor:spherical_and_sdp} with respect to both $m$ and $n$.

\begin{example}
    For this example we first consider an infinitely tall matrix $A$.
    Let the rows of $A$ consist of all unit vectors in $\R^n$.
    By taking the vector discrepancy solution $\Sigma = I_n$ and using~\eqref{eq:sdp_useful}, we see that
    \begin{align*}
        \discv(A)
        &\le \sqrt{\max_{a \in S^{n-1}}{\langle a, \Sigma a \rangle}}
        = 1\,.
    \end{align*}
    (In fact it is not hard to see that this is an equality.)
    However, for any $x \in \sqrt n \, S^{n-1}$ we have
    \begin{align*}
        \max_{a \in S^{n-1}}{\abs{\langle a, x \rangle}}
        &= \norm x_2
        = \sqrt n\,,
    \end{align*}
    which shows that $\discs(A) = \sqrt n$. Thus, the spherical discrepancy can be much larger than vector discrepancy\@.
    
    Although this example used an infinitely tall matrix $A$, we can modify it by taking the rows of $A$ to consist of a $1/2$-net of $S^{n-1}$.
    Then, the number of rows of $A$ can be taken to be $m = \exp(\Theta(n))$, and a standard argument involving nets~\citep[see the proof of][Lemma 4.4.1]{vershynin2018highdimprob} shows that we still have $\discv(A) \le 1$ and $\discs(A) \gtrsim \sqrt n$. In particular, this shows that the bound
    \begin{align*}
        \discs(A)
        &= O(\sqrt{n \wedge \log m}) \discv(A)
    \end{align*}
    obtained in Corollary~\ref{cor:spherical_and_sdp} is sharp with respect to both $m$ and $n$.
\end{example}

The next example shows that Corollary~\ref{cor:spherical_and_sdp} cannot be reversed, even with a constant depending on $m$ and $n$, so that vector discrepancy cannot in general be controlled by spherical discrepancy.

\begin{example}
\label{ex:zero_sdisc_nonzero_vdisc}
    Let $v \in S^{n-1}$ be a unit vector and suppose that the rows of $A$ consist of all vectors in
    \begin{align*}
        \mc A
        &\deq S^{n-1} \cap v^\perp\,,
    \end{align*}
    the set of unit vectors which are orthogonal to $v$ (as in the preceding example, this example can also be modified to a matrix with finitely many rows).
    Then, $\discs(A) \le \sqrt n \max_{a \in \mc A}{\abs{\langle a, v\rangle}} = 0$.
    On the other hand, we claim that $\discv(A) > 0$ for most choices of $v$.
    Suppose to the contrary that there exist unit vectors $u_1,\dotsc,u_n \in \R^n$ witnessing the fact that $\discv(A) = 0$.
    Then, for each $k\in [n]$,
    \begin{align*}
        0
        &= \max_{a \in \mc A}{\Bigl\lVert \sum_{j=1}^n a_j u_j \Bigr\rVert_2}
        \ge \max_{a \in \mc A}{\Bigl\lvert \sum_{j=1}^n a_j u_j[k] \Bigr\rvert}
        = \max_{a \in \mc A}{\abs{\langle a, u_\bullet[k] \rangle}}
        = \norm{\proj_{v^\perp} u_\bullet[k]}_2\,,
    \end{align*}
    where $u_\bullet[k] \in \R^n$ denotes the vector ${(u_j[k])}_{j=1}^n$.
    The inequality implies that $u_\bullet[k]$ is a multiple of $v$, i.e., there exists a scalar $c_k \in \R$ such that $u_j[k] = c_k v_j$ for all $j \in [n]$.
    Since $u_j$ is a unit vector,
    \begin{align*}
        1
        &= \norm{u_j}_2^2
        = v_j^2 \sum_{k=1}^n c_k^2\,, \qquad \text{for all}~j\in [n]\,.
    \end{align*}
    In order for this to hold, each coordinate of $v$ must have the same magnitude, i.e.\ $v$ must be a signing (scaled by $n^{-1/2}$).
    Hence, for most directions $v$ we have $\discs(A) = 0$ but $\discv(A) > 0$.
    
    This shows in particular that there does not exist $C > 0$ such that $\discv(A) \le C \discs(A)$ for all $m\times n$ matrices $A$, even if the constant $C$ is allowed to depend on $m$ and $n$.
\end{example}

The two preceding examples show that the statements $\discv \le C_1 \discs$ and $\discs \le C_2\discv$ do \emph{not} hold with universal constants $C_1, C_2 > 0$. In particular, since $\discv \vee \discs \lesssim \discG$ by Corollary~\ref{cor:comparison_between_disc}, it also implies that the statements $\discG \le C_1' \discv$ and $\discG \le C_2' \discs$ also do \emph{not} hold with universal constants $C_1'$ and $C_2'$.

We give another example to show that in general, Gaussian discrepancy can indeed be smaller than combinatorial discrepancy.

\begin{example}\label{ex:discG_can_be_smaller}
    Consider the case when $m=1$, so that $A$ consists of a single row, and assume that the entries of $A$ are i.i.d.\ standard Gaussians.
    Then, the discrepancy of $A = (a_1,\dotsc,a_n)$ is usually referred to as the \emph{number balancing} problem.
    
    Let us suppose that $n$ is a multiple of $3$, and divide the $n$ coordinates into $3$ groups consisting of $n/3$ coordinates each. Due to concentration of measure, the sums \[ \ell_1 \deq \sum_{i=1}^{n/3} {\abs{a_i}}\,, \qquad \ell_2 \deq \sum_{i=n/3+1}^{2n/3} {\abs{a_i}}\,, \qquad \text{and}\qquad \ell_3 \deq \sum_{i=2n/3+1}^{n} {\abs{a_i}} \] will concentrate around their expected value $\sqrt{2/\pi} \, n/3$. In particular, with high probability we will have $2\max_{i=1,2,3} \ell_i \le \sum_{i=1}^3 \ell_i$, which says that $\ell_1$, $\ell_2$, and $\ell_3$ form the side lengths of a triangle in $\R^2$. If we let $u_1, u_2, u_3 \in \R^2$ denote unit vectors corresponding to the sides of these triangles, then
    \begin{align*}
        0
        &= \sum_{i=1}^{n/3} {\abs{a_i} \, u_1} + \sum_{i=n/3+1}^{2n/3} {\abs{a_i} \, u_2} + \sum_{i=2n/3+1}^{n} {\abs{a_i} \, u_3} \\
        &= \sum_{i=1}^{n/3} {a_i \, (\sgn a_i) \,u_1} + \sum_{i=n/3+1}^{2n/3} {a_i \, (\sgn a_i) \, u_2} + \sum_{i=2n/3+1}^{n} {a_i \, (\sgn a_i) \, u_3}\,.
    \end{align*}
    This shows that with high probability, $\discv_2(A) = 0$ and hence $\discG(A) = 0$ by Proposition~\ref{prop:discG_low_rank}.
    
    On the other hand, it is well-known that $\disc(A) = \Theta(\sqrt n \, 2^{-n}) > 0$ \citep[see, \textit{e.g.},][]{KarKarLue86,Cos09,TurMekRig20}. In particular, there does not exist a constant $C > 0$ (even if the constant is allowed to depend on $m$ and $n$) such that $\disc(A) \le C \discG(A)$.
\end{example}

In summary, if we require the constants to be universal, then in general none of the inequalities in Corollary~\ref{cor:comparison_between_disc} can be reversed, and furthermore vector discrepancy and spherical discrepancy are incomparable.

In this section, we have argued that vector discrepancy and spherical discrepancy capture distinct aspects of the original discrepancy problem, and can therefore be viewed as complementary.
It is then natural to ask, whether a \emph{combination} of vector discrepancy and spherical discrepancy can control Gaussian discrepancy.
This is also motivated by the results of~\citet{Niko13} and \citet{jonesmcpartlon2020spheericaldisc} that prove the Koml\'os conjecture for vector discrepancy and spherical discrepancy respectively; hence, a control on the Gaussian discrepancy in terms of these two notions would imply the Koml\'os conjecture for Gaussian discrepancy as well.
We however resolve this question negatively via the following example.

\begin{example}
    Let the rows of $A$ consist of all unit vectors in $\R^n$ orthogonal to $e_1$.
    Then, similarly to the previous examples, we have $\discv(A) \le 1$ and $\discs(A) = 0$.
    On the other hand, for any feasible Gaussian coupling $g$,
    \begin{align*}
        \E\norm{Ag}_\infty
        = \E\max_{a \in S^{n-1} \cap e_1^\perp}{\abs{\langle a, g \rangle}}
        = \E\norm{\proj_{e_1^\perp} g}_2
        &\overset{(\star)}{\gtrsim} \sqrt{\E[\norm{(0, g_2,\dotsc,g_n)}_2^2]}
        = \sqrt{n-1}\,,
    \end{align*}
    where the inequality $(\star)$ uses Lemma~\ref{lem:log_concave}.
    Hence, $\discG(A) \gtrsim \sqrt{n-1}$.
    This in fact shows that there is no non-trivial function $f : \R_+^2 \to \R_+$ which is independent of both $m$ and $n$, such that
    \begin{align*}
        \discG(A) \le f\bigl( \discv(A), \discs(A)\bigr)
    \end{align*}
    for all matrices $A$.
\end{example}

\section{Gaussian fixed-point walk in higher dimensions}\label{scn:banaszczyk}

\subsection{High-level overview}

We begin by summarizing the Gaussian fixed-point walk as introduced in~\citet{LiuSahSawhney21}. Recall the setting of online discrepancy minimization with oblivious adversary: the adversary chooses in advance vectors $v_1,\dotsc,v_T \in \R^m$ and at each time step $t \in [T]$, the vector $v_t$ is revealed to the algorithm, upon which it must then choose a sign $\sigma_t \in \{\pm 1\}$.
The loss incurred by the algorithm is the maximum discrepancy $\max_{t\in [T]}{\norm{\sum_{s=1}^t \sigma_s v_s}_\infty}$.

Let $w_t \deq w_0 + \sum_{s=1}^t \sigma_s v_s$ denote the partial sum vector with initialization $w_0 \sim \NN(0, I_m)$. The idea behind the Gaussian fixed-point walk is to choose the signs $\sigma_t$ to ensure that $w_t \sim \normal(0, I_m)$ for all $t\in [T]$ (observe that if this holds, then the loss $\max_{t \in [T]}{\norm{w_t - w_0}_\infty}$ incurred by the algorithm is easily controlled via a union bound).
Towards this end, we may assume as an inductive hypothesis that $w_{t-1} \sim \normal(0, I_m)$.
Then, we may write
\begin{align*}
    w_t
    &= {\underbrace{\proj_{v_t^\perp} w_{t-1}}_{\sim \NN(0, I_m-v_tv_t^\perp)}} + {[{\underbrace{\proj_{v_t} w_{t-1}}_{\sim \NN(0, v_tv_t^\perp)}} + \sigma_t v_t]}\,.
\end{align*}
If we choose $\sigma_t$ independently of $\proj_{v_t^\perp} w_{t-1}$, then the two terms in the above decomposition are independent. Moreover, if we could choose $\sigma_t$ such that $\proj_{v_t} w_{t-1} + \sigma_t v_t \sim \normal(0, v_tv_t^\perp)$, then $w_t \sim \NN(0, I_m)$. These considerations lead to the problem of finding a one-dimensional Markov chain (on $\spn v_t$) whose increments belong to $\{\pm v_t\}$, and whose stationary distribution is $\NN(0, v_tv_t^\perp)$. However, a parity argument given in~\citet{LiuSahSawhney21} shows that such a Markov chain does not exist, even if we allow for changing the variance of the Gaussian from $1$.

To deal with this issue,~\citet{LiuSahSawhney21} show that such Markov chains exist if we allow the increments to take values in $\{0, \pm v_t\}$ or $\{\pm v_t, 2v_t\}$, which leads to the algorithm outputting partial or improper colorings. If the variance of the Gaussian is set to be sufficiently large, then they can further argue that the algorithm outputs an actual signing with high probability, but the large variance of the Gaussian prevents them from establishing an online Banaszczyk result.

Our contribution in this section lies in adapting the idea above to online Gaussian discrepancy, by constructing a Markov chain in $\R^r,\,r \geq 2$ whose stationary distribution is a Gaussian, and whose increments lie in on the unit sphere; this is given in Section~\ref{scn:construction_mc}. The algorithm and analysis are then given in Section~\ref{scn:alg_analysis}.

\subsection{Construction of the Markov chain}\label{scn:construction_mc}

Our goal in this section is to construct a Markov chain on $\R^r$ for $r \geq 2$ whose increments are unit vectors and whose stationary distribution is $\NN(0, \sigma^2 I_r)$ for some $\sigma^2$. In this and subsequent sections, $\sigma$ is used to denote the standard deviation (and in particular should not be confused with a signing).

The density of the $\chi$-distribution with $r$ degrees of freedom is
\[
\chi_{r}(s) = \frac{1}{2^{r/2-1} \, \Gamma(r/2)} \, s^{r - 1} e^{-\frac{s^2}{2}} \one\{s\ge 0\}\,.
\]
Hence, the density of $\norm{g}_2$ where $g \sim \NN(0, \sigma^2 I_r)$ is
\[
\chi_{r, \sigma^2}(s) 
= \frac{1}{2^{r/2-1} \, \Gamma(r/2)} \, \frac{s^{r - 1}}{\sigma^r} \, e^{-\frac{s^2}{2\sigma^2}} \one\{s\ge 0\}\,.
\]
Let $x \in \R^r$, and define the following sets:
\begin{align*}
    S_x
    &\deq \{y \in \R^r : \norm{x-y}_2 = 1~\text{and}~\norm y_2 = \norm x_2\}\,, \\
    S_x'
    &\deq \{y \in \R^r : \norm{x-y}_2 = 1~\text{and}~\norm y_2 = 1-\norm x_2\}\,.
\end{align*}
Since $\norm y_2^2 = \norm x_2^2 + \norm{x-y}_2^2 + 2 \, \langle y-x, x \rangle$, if $\norm{y-x}_2 = 1$ then from the Cauchy-Schwarz inequality, we obtain
\begin{align*}
    {(\norm x_2 - 1)}^2 \le \norm y_2^2 \le {(\norm x_2 + 1)}^2\,.
\end{align*}
Conversely, if $s \ge 0$ satisfies ${(\norm x_2 - 1)}^2 \le s^2 \le {(\norm x_2 + 1)}^2$, then there exists $y \in \R^r$ such that $\norm{y-x}_2 =1$ and $\norm y_2 = s$. From this, we deduce that
\begin{align*}
    \bigl[S_x \ne \varnothing\quad\text{if}\quad\norm x_2 \ge \frac{1}{2}\bigr] \qquad\text{and}\qquad \bigl[S_x' \ne \varnothing\quad\text{if}\quad\norm x_2 \le 1\bigr]\,.
\end{align*}
Further note that if $0 < \norm{x}_2 \leq 1$, then $S'_x$ consists of a single element given by $y = (\frac{\norm{x}_2-1}{\norm{x}_2}) x$, while $S'_0=S^{r-1}$ is the unit sphere.

The Markov chain transitions from $x\in\R^r$ with norm $\norm x_2 = s$ to the next point according to the following rules. 
\begin{enumerate}
    \item If $0 \leq s < \frac{1}{2}$, move to a point chosen uniformly at random in $S_x'$.
    \item If $\frac{1}{2} \leq s < 1$:
    \begin{enumerate}
        \item with probability $\frac{\chi_{r,\sigma^2}(1-s)}{\chi_{r,\sigma^2}(s)}$, move to a point chosen uniformly at random in $S_x'$.
        \item with probability $1-\frac{\chi_{r,\sigma^2}(1-s)}{\chi_{r,\sigma^2}(s)}$, move to a point chosen uniformly at random in $S_x$.
    \end{enumerate}    
    \item If $s \geq 1$, then move to a point chosen uniformly at random in $S_x$.
\end{enumerate}
Symbolically, this defines a transition kernel $q(x, \cdot)$ described explicitly as follows. Let $\mathsf{P}_x$ (respectively $\msf P_x'$) denote the uniform distribution on $S_x$ (respectively $S_x'$).
Then:

\begin{equation}
    \label{eqn:trans_kernel}
    q(x, \cdot) = \begin{cases}
    \msf P_x'\,, & \text{if} \, \, 0 \leq s < \frac{1}{2}  \\[0.3em]
    \frac{\chi_{r,\sigma^2}(1-s)}{\chi_{r,\sigma^2}(s)} \, \msf P_x' + (1 - \frac{\chi_{r,\sigma^2}(1-s)}{\chi_{r,\sigma^2}(s)})  \, \msf P_x\,, & \text{if} \, \, \frac{1}{2} \leq s < 1  \\
    \msf P_x\,, & \text{if} \, \, s \geq 1.
    \end{cases}
\end{equation}

This Markov chain is only well-defined if $\chi_{r,\sigma^2}(s) \geq \chi_{r,\sigma^2}(1-s)$ for all $0 \le s < 1/2$.
The following lemma shows that this holds if and only if $\sigma \ge \frac{1}{2\sqrt{r-1}}$.

\begin{lemma}\label{lem:sigma}
    The inequality $\chi_{r,\sigma^2}(s) \geq \chi_{r,\sigma^2}(1-s)$ holds for all $s \in [0,\frac{1}{2}]$ if and only if $\sigma \ge \frac{1}{2 \sqrt{r-1}}$.
\end{lemma}
\begin{proof}
    Recall that $\chi_{r,\sigma^2}(s) \propto s^{r-1} \exp(-s^2/(2\sigma^2))$, so that the mode is at $\sigma \sqrt{r-1}$ (this can be derived via elementary calculus, since the density is log-concave).
    In particular, if $\sigma < 1/(2 \sqrt{r-1})$, then the mode lies in $(0, 1/2)$, so the desired property cannot hold.

    Conversely, suppose $\sigma \ge 1/(2 \sqrt{r-1})$.
    We want to show, for all $s \in [0,1/2]$,
    \begin{align*}
        1
        &\overset{!}{\ge} \frac{\chi_{r,\sigma^2}(s)}{\chi_{r,\sigma^2}(1-s)}
        = {\bigl( \frac{s}{1-s} \bigr)}^{r-1} \exp\bigl( - \frac{s}{\sigma^2} + \frac{1}{2\sigma^2} \bigr)\,.
    \end{align*}
    Taking logarithms and rearranging, we want
    \begin{align*}
        \psi(s)
        &\deq (r-1) \log \frac{s}{1-s} - \frac{s}{\sigma^2}
        \overset{!}{\le} - \frac{1}{2\sigma^2}\,.
    \end{align*}
    Differentiating, we obtain
    \begin{align*}
        \psi'(s)
        &\ge (r-1) \, \bigl( \frac{1}{s \, (1-s)} - 4 \bigr)
        \ge 0\,,
    \end{align*}
    so $\psi$ is increasing.
    Hence,
    \begin{align*}
        \sup_{[0, 1/2]}\psi
        &= \psi\bigl( \frac{1}{2} \bigr)
        = - \frac{1}{2\sigma^2}
    \end{align*}
    as desired.
\end{proof}

With this result in hand, we verify that the Markov chain construction has the correct stationary distribution.

\begin{proposition}
\label{prop:stationary}
The Markov transition kernel $q$ on $\R^r$ defines a Markov chain with unit length increments and stationary distribution $\NN(0, \sigma^2 I_r)$ provided that $\sigma \ge \frac{1}{2\sqrt{r-1}}$.
\end{proposition}

\begin{proof}
    We represent a point in $\R^r$ in spherical coordinates, i.e., as a pair $(s, \theta)$ with $s\ge 0$ and $\theta \in S^{r-1}$.
    Let $\pi \deq \chi_{r,\sigma^2} \otimes \lambda$, where $\lambda$ is the uniform distribution on $S^{r-1}$.
    Via a change of variables, it suffices to check that for every bounded measurable function $f : \R_+ \times S^{r-1} \to \R$, it holds that
    \begin{align*}
        \iint f(s', \theta') \, \pi(\D s, \D \theta) \, q\bigl((s,\theta), (\D s',\D \theta')\bigr)
        = \int f \, \D \pi\,.
    \end{align*}
    We rewrite the LHS of this equation as
    \begin{align*}
        &\iint f(s', \theta') \, \pi(\D s, \D \theta) \, q\bigl((s,\theta), (\D s',\D \theta')\bigr) \\
        &\qquad = \int_0^\infty \Bigl[\iint f(s', \theta') \, \lambda(\D \theta) \, q\bigl((s,\theta), (\D s',\D \theta')\bigr)\Bigr] \, \chi_{r,\sigma^2}(s) \, \D s \\
        &\qquad = \Bigl(\int_0^{1/2} + \int_{1/2}^1 + \int_1^\infty\Bigr) \cdots \chi_{r,\sigma^2}(s) \, \D s
        = {\rm I} + {\rm II} + {\rm III}\,,
    \end{align*}
    and examine each of the three terms separately.
    
    For the first term,
    \begin{align*}
        {\rm I}
        &= \int_0^{1/2} \Bigl[\iint f(1-s, \theta') \, \lambda(\D \theta) \, T_\# \msf P_{(s,\theta)}'(\D \theta')\Bigr] \, \chi_{r,\sigma^2}(s) \, \D s\,,
        \intertext{where $T : \R^r \to S^{r-1}$ is the map $x \mapsto x/\norm x_2$.
        Due to the rotational symmetry we have $\int \lambda(\D \theta) T_\# \msf P'_{(s, \theta)}(\cdot) = \lambda (\cdot)$, so that evaluating the integral over $\theta$ yields:}
        &= \int_0^{1/2} \Bigl[\int f(1-s, \theta') \, \lambda(\D \theta') \Bigr] \, \chi_{r,\sigma^2}(s) \, \D s
        = \int_{1/2}^1 \Bigl[ \int f(s, \theta') \, \lambda(\D \theta') \Bigr] \, \chi_{r,\sigma^2}(1-s) \, \D s\,.
    \end{align*}

    For the second term, a similar argument yields
    \begin{align*}
        {\rm II}
        &= \int_{1/2}^1 \Bigl[\frac{\chi_{r,\sigma^2}(1-s)}{\chi_{r,\sigma^2}(s)} \iint f(1-s, \theta') \, \lambda(\D \theta) \, T_\# \msf P_{(s,\theta)}'(\D \theta') \\
        &\qquad\qquad\qquad{} + \bigl(1 - \frac{\chi_{r,\sigma^2}(1-s)}{\chi_{r,\sigma^2}(s)}\bigr) \iint f(s, \theta') \, \lambda(\D \theta) \, T_\# \msf P_{(s,\theta)}(\D \theta')\Bigr] \, \chi_{r,\sigma^2}(s) \, \D s \\
        &= \int_{1/2}^1 \Bigl[\chi_{r,\sigma^2}(1-s) \int f(1-s, \theta') \, \lambda(\D \theta') + \bigl(\chi_{r,\sigma^2}(s) - \chi_{r,\sigma^2}(1-s)\bigr) \int f(s, \theta') \, \lambda(\D \theta')\Bigr] \, \D s \\
        &= \int_0^{1/2} \Bigl[ \int f(s, \theta') \, \lambda(\D \theta') \Bigr] \, \chi_{r,\sigma^2}(s) \, \D s \\
        &\qquad\qquad\qquad{} + \int_{1/2}^1 \Bigl[ \int f(s, \theta') \, \lambda(\D \theta')\Bigr] \, \bigl(\chi_{r,\sigma^2}(s) - \chi_{r,\sigma^2}(1-s)\bigr) \,\D s\,.
    \end{align*}
    
    Finally, the third term is
    \begin{align*}
        {\rm III}
        &= \int_1^\infty \Bigl[\iint f(s, \theta') \, \lambda(\D \theta) \, T_\# \msf P_{(s,\theta)}(\D \theta')\Bigr] \, \chi_{r,\sigma^2}(s) \, \D s
        = \int_1^\infty \Bigl[\int f(s, \theta') \, \lambda(\D \theta') \Bigr] \, \chi_{r,\sigma^2}(s) \, \D s\,.
    \end{align*}
    
    Putting it together, we obtain
    \begin{align*}
        {\rm I} + {\rm II} + {\rm III}
        &= \int_0^\infty \Bigl[\int f(s, \theta') \, \lambda(\D \theta') \Bigr] \, \chi_{r,\sigma^2}(s) \, \D s
        = \int f \, \D \pi
    \end{align*}
    which verifies the stationarity of $\pi$.
\end{proof}
For the details of how to sample from the Markov kernel \eqref{eqn:trans_kernel} see Appendix \ref{scn:sampling MC}. 

\subsection{Algorithm and analysis}\label{scn:alg_analysis}

We develop some notation needed for our algorithm. Let $\sigma_{\star,r} = 1/(2 \sqrt{r-1})$; recall that this is the smallest standard deviation parameter so that for $\sigma^2 \geq \sigma_{\star,r}^2$, the transition probabilities \eqref{eqn:trans_kernel} are well-defined. For $\sigma^2 \geq \sigma^2_{\star,r}$, we let $q_{\sigma^2}$ denote the corresponding Markov transition kernel with stationary distribution $\NN(0, \sigma^2 I_r)$ as shown in Proposition~\ref{prop:stationary}. 
Also, we let $\normal(0, \sigma^2 I_{m\times r})$ denote the law of an $m\times r$ matrix whose entries are i.i.d.\ $\normal(0, \sigma^2)$.

\begin{algorithm}[H]
\caption{Rank-$r$ Gaussian Fixed-Point Walk}\label{algo:rgfpw}
\textbf{Input} $v_1, \ldots, v_T \in \mathbb{R}^m$ of $\ell_2$-norm at most $1$ \\
\textbf{Initialize} $W_0 \sim \NN(0, \sigma^2_{\star,r} \, I_{m \times r})$ \\
\For{$t=1,\dotsc,T$ }{
Define $V_t^{(k)} \in \R^{m \times r}$ to be the matrix whose $k$-th column is $v_t$ and remaining entries are $0$. \\
$\sigma_t \leftarrow \sigma_{\star,r}/\norm{v_t}_2$ \\
$z_t \leftarrow \frac{1}{\norm{v_t}_2^2} W^\T_{t-1} v_t = \frac{1}{\norm{v_t}_2^2} \, {( \langle W_{t-1}, \, V_t\rp{k} \rangle)}_{k \in [r]}$ \\
$z_t' \leftarrow$ sample from $q_{\sigma_t^2}( z_t , \cdot)$ \\
$u_t \leftarrow z_t' - z_t$\\
$W_t \leftarrow W_{t-1} + v_t u_t^\T = W_{t-1} + \sum_{k=1}^r u_t[k] \, V_t^{(k)}$
}
\textbf{Output} $u_1, \ldots, u_T \in \mathbb{R}^r$ 
\end{algorithm}

\begin{proposition}
\label{prop:online_distribution}
If $\max\limits_{t \in [T]}{\norm{v_t}_2} \leq 1$, then for all $t \in [T]$, the partial sum $W_t \in \R^{m \times r}$ of Algorithm~\ref{algo:rgfpw} is distributed as $\NN(0, \sigma^2_{\star,r} \, I_{m \times r})$.
\end{proposition}

\begin{proof}
We prove this by induction, so suppose that $W_{t -1} \sim \NN(0, \sigma^2_{\star,r} \, I_{m \times r})$.
Let $S_t$ be the subspace $S_t = \spn\{ V_t\rp{k} : k \in [r]\}$ and observe that
\[
W_{t -1} = \proj_{S_t^\perp} W_{t - 1} + \proj_{S_t} W_{t - 1} = \proj_{S_t^\perp} W_{t - 1} + \sum_{k = 1}^r z_t[k] \, V_t^{(k)}\,.
\]
Using standard facts about Gaussian random variables, these two components are independent and $z_t \sim \NN(0, \sigma_t^2 I_r)$ (recall the definition of $\sigma_t^2 \deq \sigma^2_{\star,r}/\norm{v_t}_2^2$). Also, the sampling $z_t' \sim q_{\sigma_t^2}( z_t , \cdot)$ is well-defined because $\sigma_t \geq \sigma_{\star,r}$ by the assumption that $\norm{v_t}_2 \leq 1$. Next,
\begin{align*}
   W_t &= W_{t - 1} + \sum_{k = 1}^r u_t[k] \, V_t\rp{k}
   = \underbrace{\proj_{S_t^\perp} W_{t - 1}}_{=: W_t\rp{1}} + \underbrace{\sum_{k=1}^r (z_t + u_t)[k] \, V_t\rp{k}}_{=: W_t\rp{2}}\,.
\end{align*}
Note that $z_t' = z_t + u_t \sim \NN(0, \sigma_t^2 I_r)$ and $\norm{u_t}_2 = 1$ by Proposition \ref{prop:stationary}. Hence, $W_t\rp{2}$ is distributed as a centered Gaussian on $S_t$ with variance $\sigma^2_{\star,r}$. Moreover, $W_t\rp{1}$ is independent of $W_t\rp{2}$: by the construction of our Markov chain, the distribution of $z_t' = z_t + u_t$ is a function only of $z_t$, and $z_t$ is independent of $W_t\rp{1}$ as noted above. Thus $W_t \sim \NN(0, \sigma^2_{\star,r} \, I_{m \times r})$.
\end{proof}
The following guarantee for Algorithm~\ref{algo:rgfpw} is an easy consequence. For a matrix $M \in \R^{m \times r}$, note that the $2\to\infty$ norm satisfies $\norm{M}_{2\to\infty} = \max_{1 \leq i \leq m} \norm{M_{i,:}}_2$.

\begin{theorem}\label{thm:analysis_of_alg}
    If $\max\limits_{t \in [T]}{\norm{v_t}_2} \leq 1$, then Algorithm~\ref{algo:rgfpw} outputs $u_1,\dotsc,u_T \in \R^r$ with the following properties.
    In expectation,
    \begin{align*}
        \E\max_{t\in [T]}{\Bigl\lVert \sum_{s=1}^t v_s u_s^\T \Bigr\rVert_{2\to\infty}}
        &\le \sqrt{\frac{2\log(mT)}{r-1}} + \sqrt{\frac{r}{r-1}}\,.
    \end{align*}
    Also, with probability at least $1-\delta$,
    \begin{align*}
        \max_{t\in [T]}{\Bigl\lVert \sum_{s=1}^t v_s u_s^\T \Bigr\rVert_{2\to\infty}}
        &\le \sqrt{\frac{2\log(2mT/\delta)}{r-1}} + \sqrt{\frac{r}{r-1}}\,.
    \end{align*}
    Finally, the time complexity of Algorithm~\ref{algo:rgfpw} is $O(mrT)$.
\end{theorem}
\begin{proof}
    Note that $\sum_{s=1}^t v_s u_s^\T = W_t - W_0$, where $W_0, W_t \sim \normal(0, \sigma_{\star,r}^2 \, I_{m\times r})$ by Proposition~\ref{prop:online_distribution}, so that $\norm{\sum_{s=1}^t v_s u_s^\T}_{2\to\infty} \le \norm{W_0}_{2\to\infty} + \norm{W_t}_{2\to\infty}$.
    Since $\norm\cdot_{2\to\infty}$ is the maximum $\ell_2$ norm of a row, then $\max_{t\in [T]}{\norm{W_t}_{2\to\infty}}$ can be written as a maximum over $mT$ (not necessarily independent) random variables $Y_i$, where each $Y_i$ is a norm of a $\normal(0, \sigma_{\star,r}^2 \, I_r)$ random variable. On the other hand, Gaussian concentration of Lipschitz functions~\citep[see Theorem 5.5 of][]{boucheronlugosimassart2013concentration} implies that each $Y_i \in \mss{subG}(\sigma_{\star,r}^2)$, whereas $\E Y_i \le \sigma_{\star, r} \sqrt r$.
    A similar statement holds for $W_0$.
    
    For the expectation bound, a standard sub-Gaussian bound yields
    \begin{align*}
        \E\max_{t\in [T]}{\Bigl\lVert \sum_{s=1}^t v_s u_s^\T \Bigr\rVert_{2\to\infty}}
        &\le 2 \E\max_{i\in [mT]} Y_i
        \le 2 \E \max_{i\in [mT]}{(Y_i - \E Y_i)} + 2\sigma_{\star,r} \sqrt r \\
        &\le 2\sigma_{\star, r} \, \{\sqrt{2\log(mT)} + \sqrt r\}\,,
    \end{align*}
    and the result follows from $\sigma_{\star,r} = \frac{1}{2\sqrt{r-1}}$.
    The proof of the high-probability bound also uses standard facts about sub-Gausssian random variables and is omitted.
\end{proof}

As mentioned in Section \ref{sscn:results}, our online vector Koml\'{o}s result (Theorem \ref{thm:online_vector_komlos}) is an immediate consequence of Theorem \ref{thm:analysis_of_alg}. Interestingly, Theorem \ref{thm:analysis_of_alg} also implies that scaling $\sigma_{\star,r}^2 \sim \frac{1}{4r}$ of our variance parameter is asymptotically optimal in that no sequence of Markov chains, defined on $\R^r$ for $r \to\infty$ with unit length increments and Gaussian stationary distribution, can achieve a smaller asymptotic variance. Indeed, a smaller asymptotic variance would lead to an improvement of Theorem~\ref{thm:analysis_of_alg}, which would further yield an improvement of the bound $\discv(A) \le 1$ of~\citet{Niko13}. However, the example of the identity matrix shows that this is not possible. We do not know if our variance parameter $\sigma_{\star,r}^2 = \frac{1}{4 \, (r-1)}$ can be improved for any finite $r$.

We are now ready to prove our main result, the online Banaszczyk bound for Gaussian discrepancy. 

\medskip{}

\begin{proof}[Proof of Theorem \ref{thm:online_gauss_banaszczyk}]
Let $u_1, \ldots, u_T$ denote the output of Algorithm \ref{algo:rgfpw} with rank parameter $r \geq 2$, and let $\Sigma\rp{T}_{i,j} = \langle u_i, u_j \rangle_{\ell_2}$ denote the corresponding Gram matrix. As a consequence of Theorem \ref{thm:analysis_of_alg}, we have 
\begin{align}\label{eqn:vdiscr_bound} 
\max\limits_{t \in [T]}{\Bigl\lVert \sum\limits_{s=1}^t v_s u_s^\T\Bigr\rVert_{2 \to \infty}} \lesssim 1 \lor \sqrt{\frac{\log(mT/\delta)}{r}}
\end{align} 
with probability at least $1 - \delta$. If $r \geq \log(mT/\delta),$ then combining~\eqref{eqn:vdiscr_bound} with~\eqref{eqn:union_bound_useful} from the proof of Proposition~\ref{prop:discG_union_bound} implies that 
\[
    \discG(v_1, \ldots, v_T; \Sigma\rp{T}) \lesssim \sqrt{\log(2 m)} 
\]
with probability at least $1 - \delta$. If $r \leq \log(mT/\delta)$, then combining \eqref{eqn:vdiscr_bound} with~\eqref{eqn:low_rank_useful} from the proof of Proposition~\ref{prop:discG_low_rank} implies that
\[
\discG(v_1, \ldots, v_T; \Sigma\rp{T}) \lesssim \sqrt{r} \cdot \, \sqrt{\log(mT/\delta)/r} = \sqrt{\log(mT/\delta)}
\]
with probability at least $1 - \delta$. Combining the two cases concludes the proof.
\end{proof}

\appendix

\section{Equivalence of online discrepancy formulations}\label{scn:formulations equiv}

In Section \ref{sscn:online_disc} we introduced the online Gaussian discrepancy problem, and described two formulations: one in terms of correlation matrices, and one in terms of unit vectors in $\ell_2$. We now show that the two approaches are equivalent, so that the description of Problem \ref{prob:online_gaussian_discrepancy} is without loss of generality. 

It will be useful to define the \emph{Cholesky decomposition} of a PSD matrix $\Sigma \in \R^{n \times n}$. We define $\mathsf{Ch}(\Sigma)$ to be the unique matrix $L \in \R^{n \times n}$ that is lower triangular, has exactly $r = \msf{rank}(\Sigma)$ positive diagonal elements, has $n-r$ columns consisting entirely of zeros, and satisfies
\[
\Sigma = LL^\T\,.
\]
We refer the reader to \citet[Chapter 3]{Gen12} for details regarding existence and uniqueness.

\begin{proposition}
\label{prop:gaussian_online_equivalence}
The following hold. 
\begin{enumerate}
    \item Let ${(\Sigma^{(t)})}_{t \geq 1}$ be a stream of consistent (in the sense of \eqref{eqn:gaussian_consistency}) correlation matrices. Then one can output unit vectors ${(u_t)}_{t \geq 1}$ in $\ell_2$ online, such that $\Sigma^{(t)}$ is the Gram matrix of $u_1, \dots, u_t$ for each $t \geq 1$. Moreover, only the first $t$ coordinates of $u_t$ need be non-zero. 
    \item Let ${(u_t)}_{t \geq 1}$ be a stream of unit vectors in $\ell_2$. Then one can output consistent (in the sense of \eqref{eqn:gaussian_consistency}) correlation matrices ${(\Sigma^{(t)})}_{t \geq 1}$ online, such that $\Sigma^{(t)}$ is the Gram matrix of $u_1, \dots, u_t$ for each $t \geq 1$. 
\end{enumerate}
\end{proposition}

\begin{proof}
Note that the second assertion is immediate after defining $\Sigma^{(t)}_{i,j} = \langle u_i, u_j \rangle_{\ell_2}$ for all $i,j \in [t]$ and $t \geq 1$. Now, we turn to the proof of the first assertion. For convenience, we will write $\iota_n$ for the inclusion map from $\R^k$ to $\R^n$ where $k \leq n$, with action $\iota_n(x_1, \dots, x_k) = (x_1, \dots, x_k, 0, \dots, 0)$. Let $\tilde u_t = {\msf{Ch}(\Sigma^{(t)})}_{t,:} \in \R^t$ be the last row of $\msf{Ch}(\Sigma^{(t)})$ for each $t \geq 1$, and let $u_t = \iota_\infty (\tilde u_t)$. We claim that the sequence ${(u_t)}_{t \geq 1}$ satisfies the requirements of the assertion. Clearly only the first $t$ coordinates of $u_t$ can be nonzero. Thus, it remains to show that $\langle u_i, u_j \rangle_{\ell_2} = \Sigma^{(t)}_{i,j}$ for all $i,j \in [t],\,t \geq 1$. 

Define $L^{(t)} = \msf{Ch}(\Sigma\rp{t})$ and $r_t = \msf{rank}(\Sigma\rp{t})$. It is sufficient to show that for all $t \geq 1$, the rows of $L\rp{t}$ are $\iota_t (\ti u_1), \ldots, \iota_t (\ti u_t)$. This is true for $t = 1$ by construction. We proceed by induction and suppose that the claim is true for round $t - 1$. Let $L'$ denote the principal minor of $L\rp{t}$ corresponding to the first $t-1$ indices. It suffices to show that $L' = L\rp{t-1}$ by the inductive assumption.

Both $L\rp{t-1}$ and $L'$ are lower triangular by construction. Next, by consistency and the definition of the Cholesky decomposition, we have
\[
L\rp{t-1}\, (L\rp{t-1})^\T = \Sigma\rp{t-1}=\Sigma\rp{t}_{[t-1] \times [t-1]} = L' \,{(L')}^\T\,. 
\] 
Therefore, $\msf{rank}(L\rp{t-1}) = \msf{rank}(L') = r_{t-1}$. Since columns of $L\rp{t}$ with non-positive diagonal elements contain all zeros, it follows that $L'$ has $r_{t-1}$ positive elements on its diagonal. Consequently, $L'$ also has $t-1 - r_{t-1}$ columns consisting of all zeros. We conclude that $L' = \msf{Ch}(\Sigma\rp{t-1}) = L\rp{t-1}$, as desired. 
\end{proof}

\section{Failure of some rounding schemes}\label{scn:failure_of_rounding}

As a preliminary, we refer to a univariate function $f : \R\to\R$ as a \textit{sign function} if 
\begin{align*}
	f(x) = \begin{cases}
		1 \quad  &\text{if } x > 0 \\
		-1 \quad &\text{if } x < 0 
	\end{cases}
\end{align*} 
and $f(0) \in \{-1, 0, 1 \}$. Given a vector $\xi \in \mathbb{R}^n$, we let $\sgn(\xi)$ denote a function given by the entrywise application of sign functions to the coordinates of $\xi$. 

Let $\Sigma \in \mc E_n$ be the correlation matrix achieving the optimal Gaussian discrepancy. In this section, we consider two ways of rounding $\Sigma$ to a signing. The Goemans--Williamson rounding \citep{GoeWil95} is defined by drawing a random vector $\xi \sim \NN(0, \Sigma)$ and setting
\[
\sigma_{\mathsf{GW}}(\Sigma) \deq \sgn(\xi) \,. 
\]
The PCA rounding is defined by taking a top eigenvector $v_1(\Sigma)$ of $\Sigma$ and setting
\begin{align*}
    \sigma_{\msf{PCA}}(\Sigma) \deq \sgn v_1(\Sigma)\,.
\end{align*}
Note that when $\Sigma = \sigma\sigma^\T$ has rank $1$, then $\sigma \in {\{\pm1\}}^n$ and both of these rounding schemes recover $\sigma$ (up to a global sign flip). 

In this section, we show that in the setting of both Spencer's theorem and the Koml\'os conjecture, these rounding schemes are not very effective. Concretely, for both of these settings, we construct a random matrix that has a planted Gaussian discrepancy solution with objective value zero. However, rounding this solution with either Goemans--Williamson or PCA rounding results in a value of the combinatorial discrepancy that is much larger than the guarantees obtained by direct application of Spencer's theorem or Banaszczyk's theorem.


First let us construct our planted Gaussian discrepancy solution. For $n \equiv 2 \pmod 4$  let $U$ denote the $n \times 2$ matrix whose $j$-th row is $u_j =  (\cos(\frac{2\pi j}{n}), \sin(\frac{2 \pi j}{n}))$, and let $\Sigma = UU^\T$. Then,
\begin{align}
\label{eqn:bad_gaussian_solution}
    \Sigma = \bs c \bs c^\T + \bs s \bs s^\T\,,
\end{align}
where $\bs c_j \deq \cos(\frac{2\pi j}{n})$ and $\bs s_j \deq \sin(\frac{2\pi j}{n})$. Given a vector $v \in \R^n$, define the shift operator $s: \R^n \to \R^n$ by $s(v) = (v_2,\dotsc,v_n, v_1)$ and let $w = (\mb 1_{n/2}, -{\mb 1_{n/2}})$. The lemma below verifies that the Goemans--Williamson and PCA roundings of $\Sigma$ are close to a vector in $\{w, s(w), \ldots, s^{n-1}(w) \}$. 

\begin{lemma}
\label{lem:Sigma_rounding}
	For $\Sigma$ defined as in $\eqref{eqn:bad_gaussian_solution}$ and $\sigma \in \{ \sigma_{\mathsf{GW}}(\Sigma), \sigma_{\mathsf{PCA}}(\Sigma) \}$, the following event holds with probability $1$: there exists $k \in \{0, \ldots, n - 1\}$, $i \neq j \in [n]$, and $a, b \in \{-2,-1, 0, 1,2\}$  such that 
	\[
	s^k(w) - \sigma = a e_i + b e_j.
	\]
\end{lemma} 

\begin{proof}
	For the Goemans--Williamson rounding, in fact it holds that  
	\begin{equation*}
		\Pr\big[\sigma_{\mathsf{GW}}(\Sigma) \in \{ w, s(w), \ldots, s^{n-1}(w) \}\big] = 1. 
	\end{equation*}
	To see this, observe that $\xi = U g$ where $g\in\R^2$ is distributed as standard normal. Thus $\sigma_{\mathsf{GW}}(\Sigma)_j = 1$ iff $\langle g,u_j\rangle > 0$. Since the points $u_j$ are equidistributed on the circle $S^1$, the claim follows by rotational invariance of the standard normal. 
	
	Let us now turn to the PCA-rounding scheme. Since $\bs c \perp \bs s$, the top eigenspace (corresponding to eigenvalue $n/2$) of $\Sigma$ is $\spn\{\bs c, \bs s\}$. Without loss of generality suppose that our rounding scheme chooses 
	\begin{equation*}
		v_1(\Sigma) = \frac{a}{\sqrt{a^2+b^2}} \,\bs c +\frac{b}{\sqrt{a^2 + b^2}} \,\bs s
	\end{equation*}
	as the top eigenvector of $\Sigma$ for some $(a,b) \neq (0,0)$. Let $\theta$ be such that
	\begin{align*}
		(\cos\theta,\sin\theta)=\frac{1}{\sqrt{a^2+b^2}} \,(a,b)\,.
	\end{align*}
	Then we have
	\begin{equation*}
		v_1(\Sigma) = \Bigl(\cos\bigl(\theta-\frac{2\pi j}{n}\bigr)\Bigr)_{j\in [n]}\,. 
	\end{equation*}

	Let $\sigma =  \sigma_{\mathsf{PCA}}(\Sigma) = \sgn(v_1(\Sigma))$. We consider two cases. First suppose that $v_1(\Sigma)_j= 0$ for some $j \in [n]$. Then it holds that $\theta \in \frac{\pi}{2} + \frac{2 \pi }{n}  \mathbb{Z}$. Let $j^* \in [n]$ denote the smallest positive integer such that $\theta + \frac{2\pi j^*}{n} \in \frac{\pi}{2} + \pi \mathbb{Z}$. Recalling that $n$ is even, we have 
	\begin{align*}
		\sigma_j = \begin{cases}
			\pm 1 &\quad \text{if}~j \in (j^*, j^* + n/2) \\
			\mp 1 &\quad \text{if}~j \in [n] \backslash [j^*, j^* + n/2],
			\end{cases}
	\end{align*}
	and $\sigma_j \in \{-1, 0, 1\}$ if $j \in \{ j^*, j^* + n/2  \}$. Here we use the properties of the cosine and sign functions.  By inspection, the conclusion of the lemma holds in this case. 
	
	If on the other hand, $v_1(\Sigma)_j \neq 0$ for all $j \in [n]$, then in fact we have
	\[
	\sigma \in \{w, s(w), \ldots, s^{n-1}(w) \},
	\]
	by the fact that $n$ is even and the properties of the cosine and sign functions. The lemma follows.
\end{proof}

Next, we construct a family of distributions $\mf p_{m,n}$ on $m \times n$ matrices such that the feasible Gaussian coupling $g \sim \NN(0, \Sigma)$ attains zero Gaussian discrepancy. Let $\Sigma^\perp = I_n - \frac{1}{\norm{\bs c}_2^2} \, \bs c \bs c^\T - \frac{1}{\norm{\bs s}_2^2} \, \bs s \bs s^\T$ and write $A \sim \mf p_{m,n}$ for the random $m \times n$ matrix with rows  
\[
g_1, \ldots, g_m \simiid \NN(0,\Sigma^\perp)\,.
\]
The rows of $A$ are isotropic Gaussians on the orthogonal complement of $\spn\{\bs c, \bs s\}$. Thus in particular $AU = 0$ almost surely so that $\discv_2(A) = 0$ and $\discG(A) = 0$ almost surely.

We prove an intermediate lemma that lower bounds the combinatorial discrepancy incurred by $\sigma \in \{w, s(w), \ldots, s^{n-1}(w) \}$.

\begin{lemma}
\label{lem:rounding_main_bound}
If $A \sim \mf p_{m,n}$,
then there exists a constant $c >0$ such that for $n$ sufficiently large and for all $\tau \geq 2$,
\[
\Pr\big[ \min_{0 \leq k \leq n-1} \norm{ A s^k(w) }_\infty \leq \tau \big] \leq n \exp\Bigl[ - m \exp\bigl( -\frac{c\tau^2}{\, n} \bigr) \Bigr] \,.
\]
\end{lemma}

\begin{proof}
We record some elementary calculations that will be useful later:
\begin{align}
    \sum_{j = 1}^n \cos^2\bigl( \frac{2 \pi j}{n} \bigr) = \sum_{j = 1}^n \sin^2\bigl( \frac{2 \pi j}{n} \bigr) &= \frac{n}{2} \label{eqn:trig1} \\
    \lim_{n \to \infty} \frac{1}{n} \sum_{j = 1}^{n/2} \cos\bigl(\frac{2 \pi j}{n}\bigr) &= 0 \label{eqn:trig2} \\
    \lim_{n \to \infty} \frac{1}{n} \sum_{j = 1}^{n/2} \sin\bigl(\frac{2 \pi j}{n}\bigr) &= \frac{1}{\pi} \label{eqn:trig3}
\end{align}
The identity in the first line is an elementary fact from complex variables, and the limits above are a consequence of Riemann integration.

Observe that for all $1 \leq k \leq n - 1$,
\[
\Bigl\lVert \sum_{i =1}^n w_i u_i^\T  \Bigr\rVert_2 = \Bigl\lVert \sum_{i =1}^n {s^k(w)}_i u_i^\T  \Bigr\rVert_2\,
\]
by the unitary invariance of the $\ell_2$-norm. Therefore, for all $1 \le i \le m$ and $1 \leq k \leq n - 1$, 
\[\langle g_i, s^k(w) \rangle \sim \NN(0, w^\T \Sigma^\perp w)\,.
\]
Next, 
\begin{align}
    \frac{1}{\norm{\bs c}_2^2}\, \langle \bs c, w \rangle^2 &= \frac{1}{n/2} \, \Bigl\{ 2 \sum_{k = 1}^{n/2} \cos\bigl(\frac{2 \pi k}{n}\bigr) \Bigr\}^2 = o(n)\,, \\
    \frac{1}{\norm{\bs s}_2^2} \,\langle \bs s, w \rangle^2 &= \frac{1}{n/2} \,\Bigl\{ 2 \sum_{k = 1}^{n/2} \sin\bigl(\frac{2 \pi k}{n}\bigr) \Bigr\}^2 = \frac{8n}{\pi^2} + o(n)\,,
\end{align}
by \eqref{eqn:trig1}, \eqref{eqn:trig2}, and \eqref{eqn:trig3}. We have that 
\[
w^\T \Sigma^\perp w =  \underbrace{\bigl(1 -  \frac{8}{\pi^2}\bigr)}_{=: c_0 > 0} \, n + o(n)\,.
\]
Thus the inner product of a row $g_i$ of $A$ with an output of either Goemans--Williamson or PCA rounding of $\Sigma$ is a centered Gaussian with variance roughly $c_0 n$. 

Next we lower bound $\norm{A\sigma}_\infty$ for $\sigma \in \{w, s(w), \ldots, s^{n-1}(w)\}$ using the first moment method. 
Define the random variable
\[
S = \sum_{k = 0}^{n - 1} \one\{\norm{A s^k(w) }_\infty \leq \tau\}
\]
that counts the number of shifts of $w$ attaining discrepancy $\tau$. The Markov inequality implies that
\begin{align}
\label{eqn:markov}
  \Pr[ \exists k: \norm{A s^k(w)}_\infty \leq \tau ] = \Pr[ S \geq 1] \leq \E S.  
\end{align}
By our work above, all of the vectors $A s^k(w)$ have the same distribution and have independent coordinates. Using this fact and standard Gaussian estimates, for $\tau \ge 2$,
\begin{align}
   \E S &= n \, \Pr[ \norm{Aw}_\infty \leq \tau  ]
   = n \, {\Pr_{g \sim \NN(0, w^\T \Sigma^\perp w)}[\abs g \leq \tau ]}^m \nonumber \\
   &\le n  \, \Bigl\{1 - \exp\Bigl( -\frac{\tau^2}{2\,\langle w, \Sigma^\perp w\rangle}\Bigr) \Bigr\}^m \nonumber \\
   &\leq n \exp\Bigl[ - m \exp\Bigl( -\frac{\tau^2}{2\,(c_0 + o(1)) \, n} \Bigr) \Bigr]\,. \label{eqn:GW_expectation_bd}
\end{align}
The lemma statement follows from \eqref{eqn:markov} and the previous display.
\end{proof}

Lemma \ref{lem:rounding_main_bound} allows us to demonstrate the poor performance of the rounding schemes considered in this section in two well-studied settings. 

\begin{proposition}
\label{prop:rounding_failure}
Let $A \sim \mf p_{m, n}$ and define $A' = \frac{1}{c \sqrt{ \log n} } \,A$ for a sufficiently large absolute constant $c > 0$. Let $\Sigma$ be defined as in \eqref{eqn:bad_gaussian_solution}, and let $\sigma \in \{\sigma_{\mathsf{GW}}, \sigma_{\mathsf{PCA}}\}$. Then 
\[
\discG(A') = 0 = \E \norm{A'g}_\infty \quad \text{if} \, \,g \sim \normal(0, \Sigma)\,,
\] 
and the following hold with high probability as $n \to \infty$: 
\begin{itemize}
    \item \emph{(Spencer setting)} If $m = \frac{n}{\log n}$, then $\max_{i\in [m], \;j\in [n]}{\abs{A'_{i,j}}} \leq 1$, and
    \begin{equation*}
        \disc(A') \lesssim \sqrt{\frac{n}{\log n}} \ll \sqrt{n} \lesssim \norm{A' \sigma}_\infty\,. 
    \end{equation*} 
    \item \emph{(Koml\'{o}s setting)} If $m = 10\log n$, then $\max_{j\in [n]} {\norm{A'_{:,j}}_2} \leq 1$, and
    \begin{equation*}
        \disc(A') \lesssim \sqrt{\log\log n} \ll \sqrt{\frac{n}{\log n}} \lesssim \norm{A' \sigma}_\infty\,. 
    \end{equation*}
\end{itemize}  
\end{proposition}

\begin{proof}
Recall that $\E \norm{A'g}_\infty = 0$ holds for $g \sim \normal(0, \Sigma)$ by construction.	
	
Next we verify the high probability bounds on the size of the entries of $A'$. The bound on the size of the largest entry of $A'$ in the Spencer setting follows from the union bound and tail bounds for the standard Gaussian distribution. Now, for the Koml\'{o}s setting, since $\Sigma^\perp \preceq I_n$, it follows that $\Var(A_{i,j}) \leq 1$ for all $i, j$. Let $\tau_j \leq 1$ denote the common variance of the entries in column $j$. We have that $\norm{A_{:,j}}_2 - \sqrt{\tau_j m}$ is sub-Gaussian with variance proxy $1$ \citep[see \textit{e.g.}][Chapter 3]{vershynin2018highdimprob}. Applying the sub-Gaussian tail bound and the union bound over $1 \leq j \leq n$ implies that $\max_{j\in [n]}{\norm{A_{:,j}}_2} \lesssim \sqrt{\log n}$, as desired.
	
Given the previous bounds on the entries of $A'$, the upper bounds on $\disc(A')$ follow immediately from Spencer's theorem and Banaszczyk's bound, respectively. To finish, we verify the lower bounds on $\norm{A' \sigma}_\infty$ when $\sigma \in \{\sigma_{\mathsf{GW}}, \sigma_{\mathsf{PCA}}\}$. By Lemma \ref{lem:Sigma_rounding}, it holds with probability $1$ that for $\sigma \in  \{\sigma_{\mathsf{GW}}, \sigma_{\mathsf{PCA}}\}$, 
\[
s^k(w)- \sigma = a e_i + b e_j
\]
for some $k \in \{0, \ldots, n-1\}, i \neq j \in [n], $ and $|a|, |b| \leq 2$. By the union bound over this event and the event that $\max_{i,j} |A'_{i,j}| \leq 1$, it holds with high probability that
\begin{align}
\label{eqn:sigma_approx_skw}
\| A' ( \sigma - s^k(w) ) \|_\infty \leq 4,
\end{align}
for some $k \in \{0, \ldots, n-1 \}$, by the triangle inequality. By the union bound again over the event \eqref{eqn:sigma_approx_skw} and the event of Lemma \ref{lem:rounding_main_bound}, it holds that
\begin{align*}
	\| A' \sigma \|_\infty \geq \min_{0 \leq k \leq n-1} \|  A' s^k(w)  \|_\infty  - 4
	\gtrsim \begin{cases}
	\sqrt{n} \quad &\text{if } m = \frac{n}{\log n} \\
	\sqrt{\frac{n}{\log n}} \quad &\text{if } m = 10 \log n \\
	\end{cases}. 
\end{align*}
Precisely, for $ m = \frac{n}{\log n}$, we set $\tau = c' \sqrt{n \log n}$ in Lemma \ref{lem:rounding_main_bound}, and for $m = 10 \log n$, we set $\tau = c'' \sqrt{n} $, for sufficiently small absolute constants $c', c'' > 0$. 

\end{proof}


Furthermore, when $m = \frac{n}{\log n}$, a signing chosen uniformly at random already achieves discrepancy $O(\sqrt{n})$ for the matrix $A'$ with high probability. This indicates that in the Spencer setting our rounding schemes $\sigma_{\mathsf{GW}}$ and $\sigma_{\mathsf{PCA}}$ fail to `beat the union bound'. We justify this below for completeness.

\begin{lemma}
With high probability over the random matrix $A \sim \mf p_{m,n}$, a uniformly random sign $\sigma \sim \operatorname{unif}(\{\pm 1\}^n)$ satisfies
\begin{equation*}
    \E[\norm{A'\sigma}_\infty \mid A'] \lesssim \sqrt n\,, 
\end{equation*}
where $A' = \frac{1}{c\sqrt{\log n} }A$ as before. 
\end{lemma}
\begin{proof}
By standard concentration results~\citep[see, \textit{e.g.},][Theorem 5.5]{boucheronlugosimassart2013concentration} and using the fact that $\Var(A_{i,j}) \leq 1$ for all $i,j$, it holds for fixed $i$ that
\begin{align*}
    \Pr[ \norm{A_{i,:}}_2 - \sqrt{n} \geq t ] \leq \exp( -t^2/2)\,.
\end{align*}
Thus the union bound implies that the event $E = \{ \max_{1 \leq i \leq m}{\norm{ A_{i,:} }_2} \leq 2\sqrt{n}\}$ holds with probability at least $1 - me^{-n/2}$. Conditionally on the matrix $A$, the random variable $\sum_{j = 1}^n A_{i, j}\sigma_j$ is sub-Gaussian with parameter $\norm{A_{i,:}}_2^2$, so the maximal inequality for sub-Gaussian random variables~\citep[see, \textit{e.g.},][Theorem 2.5]{boucheronlugosimassart2013concentration} yields
\begin{align*}
    \E[\norm{A\sigma}_\infty \mid A]
    &\le \bigl(\max_{1 \le i \le m}{\norm{A_{i,:}}_2} \bigr) \,\sqrt{2 \log(2m)}
    \le 2\sqrt{2n \log \frac{2n}{\log n}}
    = O(\sqrt{n\log n})
\end{align*}
with probability at least $1 - me^{-n/2}$, as desired.
\end{proof}

\section{Sampling from the Markov chain}\label{scn:sampling MC}

To implement the Gaussian fixed point walk in dimension $r > 1$, we need to be able to sample from the Markov kernel \eqref{eqn:trans_kernel}. Clearly, to this end it is sufficient to solve the following problem. 

\begin{problem}\label{prob:sampling}
Given a point $x \in \R^r$ and a target radius $s \geq 0$, output a uniformly random point from
\begin{equation}
    S \deq \{y \in \R^r\,:\,\norm{x-y}_2 = 1,\, \norm{y}_2=s\}\,, 
\end{equation}
provided $S \neq \varnothing$. 
\end{problem}
Let $\delta = y-x$. Equivalently, we need to sample from the set of $\delta \in \R^r$ with $\norm{\delta}_2 = 1$ and $\norm{x+\delta}_2 = s$. Expanding and rearranging the latter, we obtain
\begin{equation*}
    \bigl\langle\frac{x}{\norm{x}_2}, \delta\bigr\rangle = \frac{s^2 - \norm{x}_2^2 - 1}{2\,\norm{x}_2}\,. 
\end{equation*}
The set of such $\delta$ can be parametrized as $\delta = \lambda x/\norm{x}_2 + \sqrt{1-\lambda^2}\, w$ where $w$ is a unit vector orthogonal to $x$, and $\lambda = (s^2-\norm{x}^2_2-1)/(2\norm{x}_2)$. Uniformly sampling such a $w$ is easy: given a standard Gaussian $g \sim \normal(0, I_r)$, project it onto $x^\perp$ and normalize:
\begin{equation}\label{eqn:sampling projection step}
    w = \frac{(I_r - xx^\T/\norm{x}_2)\, g}{\norm{(I_r - xx^\T/\norm{x}_2)\, g}_2}. 
\end{equation}
Putting it all together, we obtain the following algorithm.

\begin{algorithm}[H]
\caption{Sampling from the Markov kernel \eqref{eqn:trans_kernel}} \label{algo:markov kernel}
Input: $x \in \R^r$, $s \geq 0$. \\
Set $\lambda = (s^2-\norm{x}^2_2-1)/(2\norm{x}_2)$. \\
\If{$\lambda \notin [0,1]$}
    {Output: $S$ is empty.}
\Else
    {Draw $w$ according to \eqref{eqn:sampling projection step}.\\
    Set $\delta = \lambda x/\norm{x}_2 + \sqrt{1-\lambda^2}\, w$.\\
    Output: $x+\delta$. }
\end{algorithm}

\acks{We thank Kwangjun Ahn and Chen Lu for helpful conversations. SC was supported by the Department of Defense (DoD) through the National Defense
Science \& Engineering Graduate Fellowship (NDSEG) Program. PR is supported by NSF grants IIS-1838071, DMS-2022448, and CCF-2106377.}

\bibliography{bibliography.bib}

\end{document}